\documentclass{article}

\usepackage{amsfonts}
\usepackage{amsmath}
\usepackage{amsthm}
\usepackage{amssymb}
\usepackage{verbatim}
\usepackage{xcolor}
\usepackage{enumitem}

\usepackage{mdframed}

\newmdenv[
  linecolor=black,
  linewidth=0.5pt,
  topline=true, bottomline=true,
  leftline=true, rightline=true,
  backgroundcolor=gray!5,
  skipabove=8pt, skipbelow=8pt,
  innertopmargin=6pt, innerbottommargin=6pt
]{examplebox}

\theoremstyle{definition}
\newtheorem{exinner}{Example}

\newenvironment{example}[1][]
  {\begin{examplebox}\begin{exinner}[#1]}
  {\end{exinner}\end{examplebox}}

\usepackage{graphicx} 
\usepackage{overpic}
\usepackage[final]{hyperref}

\newcommand{\comm}[1]{}
\newcommand{\ud}{\,\mathrm{d}}

\def\({\left(}
\def\){\right)}
\def\oli{\overline}

\def\raw{\rightarrow}

\def\no={\neq}

\def\E{{\mathbb E}}

\def\R{{\mathbb R}}

\def\DD{{\mathcal D}}

\def\GG{{\mathcal G}}

\def\OO{{\mathcal O}}
\def\PP{{\mathcal P}}

\def\SS{{\mathcal S}}

\def\XX{{\mathcal X}}

\def\al{\alpha}
\def\be{\beta}
\def\ga{\gamma}
\def\de{\delta}

\def\th{\theta}

\def\ka{\kappa}
\def\la{\lambda}

\def\si{\sigma}

\def\om{\omega}

\def\Ga{\Gamma}
\def\De{\Delta}

\def\La{\Lambda}

\newtheorem{Thm}{Theorem}[section]
\newtheorem{Prop}[Thm]{Proposition}
\newtheorem{Lem}[Thm]{Lemma}

\newtheorem{Cor}[Thm]{Corollary}
\newtheorem{Rem}[Thm]{Remark}
\newtheorem{Def}[Thm]{Definition}

\usepackage[
backend=biber,
sorting=none
]{biblatex}
\usepackage{graphicx}
\graphicspath{ {./images/} }

\title{Effect of antibiotic spectrum on the abundance of resistant bacteria in multispecies communities}

\date{}

\begin{document}

\maketitle

\begin{flushleft}

Magnus Aspenberg \textsuperscript{1},
Erik Andreas Martens \textsuperscript{1,2},
Kristofer Wollein Waldetoft \textsuperscript{3,4,*}
\\
\bigskip
\textbf{1} Centre for Mathematical Sciences, Lund University, Lund, Sweden. 
\\
\textbf{2}
Centre for Mathematical Modeling - Health and Disease, Institute for Mathematics and Physics, Dept. of Environment and Science, Roskilde University, Denmark
\\
\textbf{3} Department of Infectious Diseases, Uppsala University Hospital, Uppsala, Sweden.
\\
\textbf{4}  Faculty of Medicine, Department of Clinical Sciences, Division of Infection Medicine, Lund University, Lund, Sweden.
\bigskip

* kristofer.wollein.waldetoft@akademiska.se

\end{flushleft}

\vspace{5mm}

\section*{Keywords}
Antibiotic resistance, evolution, community ecology, microbiome, Lotka-Volterra, mathematical model.

\section*{}
\textbf{Antibiotic resistance is a major threat to global health. It emerges in multispecies microbial communities under antibiotic exposure. This makes antibiotic spectrum --- a drug's distribution of effects across species --- a potential key parameter in resistance management. However, current theory for resistance dynamics is not well-developed for the effect of antibiotics in a multispecies setting. Analysing a well-established model in community ecology theory --- the generalized Lotka-Volterra model --- we develop a simple mathematical measure for how an antibiotic effect on one taxon (strain or species) affects another taxon through all direct and indirect interactions in a complex interaction network. Using this measure, we predict the expected effects of different antibiotic spectra on the abundance of resistant taxa in microbial communities. This furthers our understanding of the community ecology of antibiotic resistance, and provides a formal theoretical basis for empirical work on optimal antibiotic choice.}

\section*{}
The evolution of antibiotic resistance is among the major threats to human health in our time \cite{WHO_AMR}. It is driven by antibiotics \cite{LID_Popgen}, and their prudent use is therefore a top priority \cite{ESCMID_Stew,WHO_Stew}. This has two prongs. First, and well known, is the effort to limit the amount of antibiotics used \cite{EU_Amount}. Second, and perhaps less understood, is the matter of antibiotic choice \cite{WHO_Choice}. --- Given the need for treatment, which antibiotics, or more precisely, what properties of antibiotics, minimize the evolution of resistance?

Several properties of antibiotics may influence the degree to which they drive the evolution of resistance. These include the molecular target of the drug, which may affect the likelihood of resistance mutations \cite{NatRevTarget}, and the pharmacokinetics, which determine the distribution of the drug throughout the patient's body and thus the concentrations that bacteria experience at different anatomical sites \cite{Conc_Sites}. One property, however, stands out both in clinical decision making and in terms of microbial evolutionary ecology, and this is the property of antibiotic \textit{spectrum} (defined below).

The current drive for prudent antibiotic use includes a strong preference for narrow antibiotic spectrum \cite{EU_Prudent,WHO_Spectrum}, that is, antibiotics that have a strong effect only on a limited set of bacteria. There are several potential rationales for this. One is that broad spectrum drugs are crucial in emergency situations, where the culprit pathogen is unknown, and they should therefore be preserved for this context. Another is that, acting on a smaller set of bacteria, narrower spectrum antibiotics impose selection only on this smaller set, and may therefore drive resistance in fewer bacterial species.

However, there is yet another aspect to spectrum in the evolution of resistance that is only apparent when viewed through the lens of community ecology. Most antibiotic exposure experienced by common pathogens is bystander exposure \cite{Tedijanto}, where the pathogen is part of a commensal microbial community and not the target of treatment. In such communities, bacteria of different species interact. They may inhibit or facilitate each other, or show combinations of the two, where one species gains and the other loses from the interaction, and these interactions may be stronger or weaker for different species pairs \cite{Palm_Fost}. The absolute fitness, which is the proper fitness measure for resistance emergence \cite{Day_Read}, of a resistant pathogen strain, is then affected not only by selection within the species, but by the ecological interactions among all interacting species in the microbial community. The spectrum of the drug describes its distribution of effects across species, and should therefore be a key determinant of the resistant pathogen's absolute fitness. In other words, spectrum may determine the extent to which an antibiotic drives resistance, via interspecies interactions.

In the present investigation, \textit{we use community ecology theory to assess the effect of antibiotic spectrum on the absolute fitness of a focal resistant taxon in a microbial community under bystander exposure.} The results deepen our understanding of antibiotics in microbial ecology, and tie in with current matters in antibiotic stewardship.

\section*{Results}
How does the antimicrobial spectrum of an antibiotic affect the absolute abundance of a fully resistant taxon in a multispecies microbial community? We address this question with an analysis of community ecology theory, specifically the generalised Lotka-Volterra model (gLV).

In the following, we first define the central concept of the study --- antibiotic \textit{spectrum} --- and specify the model and assumptions. We then set the stage for the analysis by putting the question into the context of a network of ecologically interacting taxa. Then we provide deterministic results aimed at an in-depth mathematical analysis of the gLV model, followed by an interpretation of those results. Then we give a probabilistic version of the results aimed at empirical applicability, again followed by an interpretation. Lastly, we discuss the biological meaning of the results, their limits and implications for antibiotic stewardship.

\subsection*{Antibiotic spectrum}
Medicine is practically oriented, and the concept of antibiotic spectrum is no exception. It roughly refers to the set of medically relevant taxa, on which the antibiotic has a strong enough effect for adequate treatment. It does not include weak effects that may matter for microbial ecology and evolution but do not suffice for reliable treatment of infection. For the purpose of clarity, we develop the following working definitions.

In our model, the \textit{ecological spectrum} of an antibiotic is a vector (list) of numbers that represent the effect of the drug on each taxon. Increasing the magnitude of at least one of these numbers, without decreasing any --- that is, increasing the strength of the effect on at least one taxon without decreasing the effect on any other taxon --- makes the ecological spectrum \textit{stronger}.

We may set a threshold value for these numbers, above which the antibiotic effect is strong enough for adequate treatment and below which it is not. Setting all numbers above threshold to $1$ and those below threshold to $0$ gives a vector that lists which taxa are treatable with the drug. We call this vector the \textit{medical spectrum} of the drug, and the number of ones, i.e., the number of treatable taxa, we call the \textit{breadth} of the medical spectrum. We see that making the ecological spectrum stronger may, but need not, make the medical spectrum broader.

The ecological spectrum is most relevant for understanding the causes of resistance evolution, whereas the medical spectrum is more relevant for its public health consequences. Studying the community ecology of resistance evolution, we compare ecological spectra that differ in strength.

\subsection*{Model}
We use the generalised Lotka-Volterra model (gLV) with $n$ taxa --- for a recent review, see reference \cite{Allesina_100}. This is a well-studied model that is commonly used for modelling microbial communities \cite{Davis_LV}. Its strengths and weaknesses are reviewed in reference \cite{gLV_Rev}. The model allows for different types of interactions among microbes --- inhibitory, neutral and facilitatory --- in all combinations. The possible dynamics depend on these interactions. They include stable equilibria and other attractors, but also unrealistic growth to infinity, which can occur if facilitatory interactions are allowed. To investigate the effect of antibiotic spectrum, we modify the  model by adding antibiotic effects, as previously done \cite{Sundius}. Formally, the model is defined as follows.

Let $x_i = x_i(t)$ be the abundance of taxa $i=1, \ldots, n$ at time $t$, where $n$ is the number of taxa, subject to antibiotic effect $l_i \geq 0$: 
\begin{equation}\label{eq:model}
   x_i'(t) = r_i x_i(t) \left(1 - \frac{1}{K_i} \sum_{j =1}^n \al_{ij} x_j(t) \right) - x_i(t) l_i,
\end{equation}
where $r_i >0$ is the intrinsic {\em growth rate} and $K_i > 0$ is the {\em carrying capacity} of taxon $i$. The interactions between taxa are described by the {\em interaction matrix} $A$, where $(A)_{ij} = \al_{ij}$. Specifically, $\al_{ij}$ measures the influence of $x_j$ on $x_i$. If $\al_{ij} > 0$, $x_j$ inhibits the growth of $x_i$, and if $\al_{ij}  < 0$, $x_j$ facilitates the growth of $x_i$. The effect of the antibiotic on $x_i$ is given by $l_i$, and the vector $\oli{l}=(l_1, \ldots, l_n)$ describes the \emph{ecological spectrum} of the drug. We let $x_1$ be the focal resistant taxon, and, hence $l_1 = 0$.

\subsection*{Assumptions} \label{sec:assumptions}
To strike a balance between scope and tractability, we make the following assumptions.
\begin{enumerate}
    \item {\textit{Pairwise microbial interactions (i.e. the coefficients $\al_{ij}$ in the interaction matrix $A$) are inhibitory, neutral or facilitatory in any combination.} This gives the study a broad scope.}        
    \item{\textit{The microbial community is at a stable equilibrium.}  This excludes unrealistic growth to infinity, and is consistent with, though not strictly implied by, the finding that the human gut microbiota shows stability over time \cite{Coyte_Foster,Faith_stable}. In Appendix \ref{app:num_algorithm} and \ref{app:math_analysis} we provide analyses of the existence of equilibria in the gLV. See Table \ref{tab:statistics_attractors_normal} for an illustration.  We only consider taxa that are present, i.e. we assume a {\em coexistence} equilibrium where all $x_i > 0$. This assumption is supported by reference \cite{Lechon-Alonso-etal}.} 
        \item{\textit{The equilibrium persists and remains stable under the addition of antibiotic or change in spectrum.} For large communities, this is supported by reference \cite{Lechon-Alonso-etal}. Recall that we are not modelling the effect of the antibiotic at an infection site, where the pathogen goes extinct under treatment, but bystander exposure at an off-target anatomical site, where drug concentration may be lower. In Appendix \ref{app:math_analysis} (Section \ref{extinction}), we relax this assumption and allow taxa to go extinct.}
    \end{enumerate}

For probabilistic, numerical and combinatorial results, we also assume the following.
    \begin{enumerate}
    \setcounter{enumi}{3}
    \item{\textit{Interactions are weaker between species than within species.} This is a natural assumption for mutually inhibitory interactions (i.e., competition) due to partial niche overlap. Empirical measurements of ecological interactions in microbial communities are scarce, but indicate that they are mostly inhibitory and weaker between than within species \cite{Palm_Fost,Coyte_Foster}. For the interaction matrix this means that $|\al_{ij}| < \al_{ii}$. By convention $\al_{ii} = 1$, and hence the assumption is $|\al_{ij}| < 1$. }
    \item{\textit{Interaction coefficients are random and \textit{i.i.d.}} (independent and identically distributed).}
    \end{enumerate}

\begin{figure}[htp!]
\centering
{\includegraphics[height=0.2\textheight]{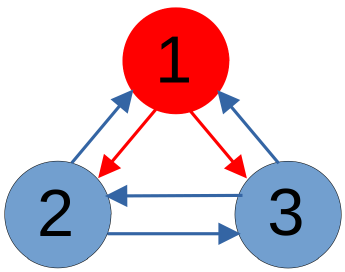}}
\caption{
Interaction network with one focal resistant taxon 1 (red) and non-resistant taxa 2 and 3 (blue).
}
\label{Graph}
\end{figure}

\subsection*{The question in the context of an interaction network}
Consider first a simple microbial community with only two taxa that interact by competition. If one taxon is killed by an antibiotic, the other taxon is relieved of competition, and increases in abundance. This is competitive release. For other types of interaction --- exploitation or mutualism --- the results change. There is, however, a straightforward pattern. If one taxon, call it $x_1$, is inhibited by the other taxon, $x_2$, it increases when $x_2$ is killed. If instead $x_1$ is facilitated by $x_2$, it decreases when $x_2$ is killed. See reference \cite{Sundius} for an analysis.

Now consider more complex communities with more interacting taxa. Figure \ref{Graph} shows the interaction network for $3$ taxa. Already with $3$ taxa the picture is much more complicated. Tracking the paths across the interaction network, there are both direct and indirect effects, as well as several feedback loops. With hundreds or thousands of taxa, one may quickly lose track.

Figure \ref{fig:Dynamics} illustrates community dynamics for such a $3$ taxa community under different interaction regimes (rows; inhibition/competition, facilitation/mutualism, mixed) and antibiotic spectra (columns; no antibiotic, weak/narrow, strong/broad). Recall that interaction coefficients are positive for inhibition and negative for facilitation. In these numerical examples, the resistant taxon increases under antibiotic exposure when interactions are inhibitory, and decreases when interactions are facilitatory, and these effects are stronger for a stronger/broader spectrum. We return to these results below. In the example with mixed interactions the effects are more complicated. In Example \ref{Example1}, we walk through the mathematics underlying this example.

In the next several sections, we generalise the analysis cited above (reference \cite{Sundius}) from a community with only two taxa and an antibiotic that affects a single taxon to communities of any size exposed to antibiotics with any spectrum. 

First we give deterministic results. These are aimed to deepen our understanding of community ecology as represented by the gLV, utilizing full information on the ecological parameters. Second, because empirical information on ecological parameters is limited, we also provide probabilistic versions of the results, aimed towards empirical applicability. Both deterministic and probabilistic results are presented with one section of formal analysis followed by a section that interprets that analysis.

\begin{figure}[htp!]
\centering
\begin{overpic}[width=\textwidth]{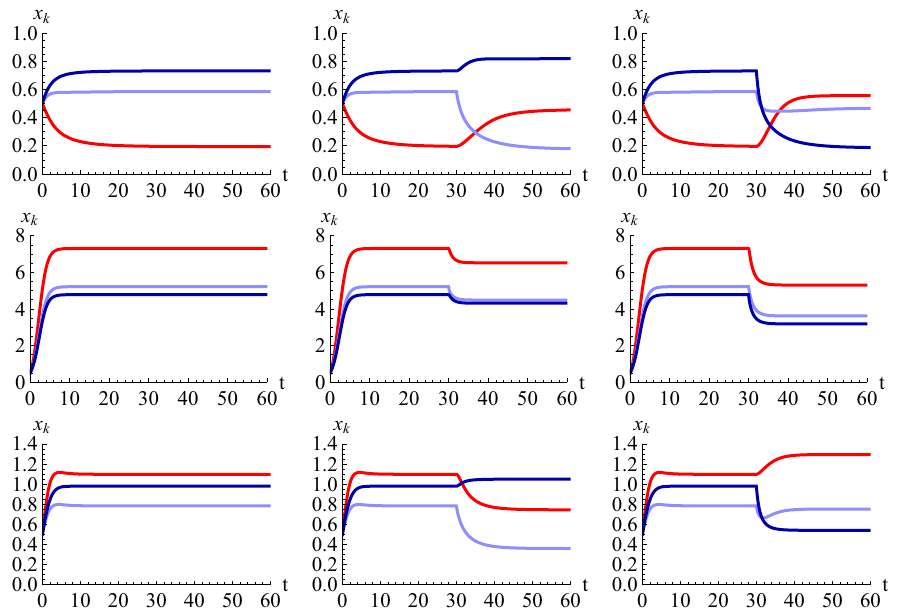}
\put(-2,66){(a)}
\put(31,66){(b)}
\put(65,66){(c)}
\put(-2,43){(d)}
\put(31,43){(e)}
\put(65,43){(f)}
\put(-2,20){(g)}
\put(31,20){(h)}
\put(65,20){(i)}
\put(12,66){\footnotesize $l=(0,0,0)$}
\put(45,66){\footnotesize $l=(0,l_2,0)$}
\put(80,66){\footnotesize $l=(0,l_2,l_3)$}
\put(-3,5.5){\rotatebox{90}{\footnotesize $\alpha_{ij}$ (mixed)}}
\put(-3,27){\rotatebox{90}{\footnotesize$\alpha_{ij}$ (negative)}}
\put(-3,50.5){\rotatebox{90}{\footnotesize $\alpha_{ij}$ (positive)}}
\end{overpic}
\caption{\label{fig:Dynamics}
Dynamics and outcomes for numerical simulations of different scenarios in a three-taxon interaction network with one focal resistant taxon~1 (red) and non-resistant taxa~2~(light blue) and~3 (dark blue). The antibiotic is introduced at time $t=30$.
Columns correspond to \emph{antibiotic spectra}: no antibiotic (a,~d,~g); antibiotic targeting taxon~2 (b,~e,~h); antibiotic targeting taxa~2 and~3 (c,~f,~i).
Rows correspond to \emph{interaction-sign regimes}: all interaction coefficients ($\alpha_{ij}$, $i \neq j$) positive (a--c); all negative (d--f); mixed (g--i).
Parameters: $l_2 = 0.3$, $l_3 = 0.5$,
$\alpha_{12}=3/4$, $\alpha_{21}=\alpha_{31}=1/4$, $\alpha_{13}=\alpha_{23}=1/2$ and $\alpha_{32}=3/8$ and initial conditions $x_1(0)=x_2(0)=x_3(0)=1/2$ (compare also with Example \ref{Example1}).
Panels~(d--f) use the negatives of all $\alpha$-values; panels~(g--i) use $\alpha$-values identical to~(a--c), except that $\alpha_{13}$, $\alpha_{23}$, and $\alpha_{32}$ are negated.
}
\end{figure}

\subsection*{Deterministic results}
Here, we provide analytic, deterministic results for how the equilibrium abundance of a fully resistant taxon changes as a function of the shift in abundance of another taxon or other taxa. Such net effects in gLV have been extensively studied. See reference \cite{Novak-etal-2026} for a review. We will, however, approach the problem slightly differently, and extend the results to antibiotic effects.  An interpretation of the results is given in the next section. 

Suppose that the interaction matrix $A$ is an invertible $n \times n$ matrix. We let $c_{ij}$ be the cofactors of $A$, and $\det(A)$ the determinant of $A$. By Corollary \ref{Cor-1} in Appendix \ref{app:math_analysis}, a matrix $A$ must have $\det(A) > 0$ if there is a coexistence stable equilibrium (see also reference \cite{Lechon-Alonso-etal}). Throughout the paper, we assume that there is a coexistence stable equilibrium (cf. Assumption 2). 

We first study the microbial community governed by the gLV system \textit{without} the loss term $l_i x_i$, i.e. the original gLV equations without antibiotic effects. 

\comm{
Since we are looking at coexistence equilibria we have 
\[
\al_{i1}x_1 + \ldots + \al_{in}x_n = K_i, \qquad i = 1, \ldots, n. 
\]
Every such equation represents a hyperplane $P_i$ in the phase space $\R^n$. The equilibrium $\oli{x}^*$ is the intersection of all these hyperplanes;
\[
\oli{x}^* = \bigcap_{i=1}^n P_i. 
\]
}
Suppose now that we choose a taxon $x_j$ and control its abundance. Then the abundance of the other taxa is a function of the abundance $x_j$. In the gLV equations this means that row $j$ is disregarded, and replaced by the simple equation $x_j = q$.
By changing $x_j$ from $q$ to something slightly bigger or smaller than $q$, the corresponding equilibrium will shift.  If we assume that we change $x_j$ by an amount $\de_j$, then the change $\de_k$ of taxon $x_k$ ($k \neq j$) will be
\[
\de_k = \frac{c_{jk}}{c_{jj}} \de_j.
\]

This simple relation follows from standard linear algebra, and its connection to interaction networks was discovered quite long ago (see \cite{Nakajima-1992} or \cite{Novak-etal-2026}, Section 4.2). We reprove it in Appendix \ref{app:math_analysis}, Theorem \ref{quotient}.  By linear superposition we also obtain: 
\begin{Thm} \label{main-eco}
    Suppose that we change the abundance of taxon $x_j$ by $\de_j$, for $j \in I$, where $I \subset [n]$ is an index set not containing $k$. Then the change $\de_k$ of taxon $x_k$ is
\[
\de_k = \sum_{j \in I} \frac{c_{jk}}{c_{jj}} \de_j.
\]
\end{Thm}

Now we study the gLV with the loss term representing antibiotic effects. 
Suppose that we have an ecological spectrum vector $\oli{l} = (l_1, \ldots, l_n)$, where $l_1 = 0$, i.e. we assume that $x_1$ is fully resistant. We study the effect of increasing $l_j$ for a \emph{single taxon} $j \neq 1$ by a small amount $\la_j > 0$, i.e. we make the ecological spectrum stronger by increasing $l_j$. 
Recall that $r_j$ is the intrinsic growth rate of $x_j$, and  $K_j$ the carrying capacity. 
Then we have the following change $\de_1$ in abundance of taxon $1$ (see  Appendix \ref{app:math_analysis} following Theorem \ref{quotient}):
\begin{equation} \label{delta-1-a}
\de_1 = - \la_j \frac{K_j}{r_j} \; \frac{c_{j1}}{\det(A)}.
\end{equation}
By linear superposition we also have the following. 

\begin{Thm} \label{main-detm}
    Suppose $A$ is an invertible $n \times n$-matrix and that there is a stable coexistence equilibrium. Suppose we change the ecological spectrum at element $l_j$ with an amount $\la_j$ for some collection $I$ of indices $j \neq 1$. Then the change in abundance of taxon $1$ is given by
    \begin{equation} \label{delta-1-b}
        \de_1 = \sum_{j \in I} -\la_j \frac{K_j}{r_j} \frac{c_{j1}}{\det(A)}. 
    \end{equation}
\end{Thm}

Hence, if we know all the interaction coefficients in $A$, the carrying capacities $K_j$ and intrinsic growth rates $r_j$, then the change of taxon $1$ is given by formula (\ref{delta-1-b}). Clearly, the formula (\ref{delta-1-b}) can be used to track any set of resistant taxa.

To illustrate these results, we provide an explicit numerical application in Example \ref{Example1}.






%

\subsection*{Interpretation of the deterministic results}
Here we walk through the theorems in the previous section. We assume the gLV system is at a stable equilibrium (see the Appendix \ref{app:num_algorithm} and \ref{app:math_analysis} for details), and track the change in abundance of a focal resistant taxon. 

Theorem \ref{main-eco} applies to the standard generalised Lotka-Volterra model in community ecology, \textit{without} antibiotic effects (and it is equivalent to a result in reference \cite{Nakajima-1992}). It describes the change in abundance of one taxon ($k$) induced by a change in the abundance of one or more other taxa ($j$), taking all direct and indirect effects across the interaction network into account. Considering first a single taxon $j$ for simplicity, the change in taxon $k$, denoted $\delta_k$, equals the change in taxon $j$, denoted $\delta_j$, times a scaling factor ($c_{jk}/c_{jj}$). We call this scaling factor the \textit{interaction quotient}. It captures the influences of taxon $j$ on taxon $k$ via all interactions in the network, or equivalently, all interaction coefficients $\alpha_{ij}$ in the interaction matrix. The sign of the quotient depends on the coefficients $\alpha_{ij}$, and it determines the direction of change in the focal taxon $k$. The interaction quotient thus describes the effect of one taxon on another in the context of a multispecies community, and condenses the community ecology involved into an interaction quotient. This is then generalised to the combined effect of changing any set of taxa except the focal taxon $k$.

Theorem \ref{main-detm} applies to our full model, the generalised Lotka-Volterra model \textit{with} antibiotic effects. The antibiotic effect is modelled as a per capita loss rate. The theorem gives the change in abundance of a focal resistant taxon, induced by a change in antibiotic effect on one or more other taxa ($j$), again taking all interactions into account. We call the focal taxon $1$, to simplify notation. Again considering first a single taxon $j$ for simplicity, the change in taxon $1$ ($\delta_1$) is given by a product of $3$ factors. First, $\lambda_j$ is the change in antibiotic effect on taxon $j$. In a scenario where one antibiotic is changed to another, as in therapy escalation or de-escalation, $\lambda_j$ denotes the difference between these antibiotics in their effect on taxon $j$. In a scenario where treatment is initiated (a change from zero to non-zero antibiotic effect), $\lambda_j$ is simply the effect of the antibiotic on taxon $j$. Second, $K_j/r_j$ is the ratio of carrying capacity to intrinsic growth rate for taxon $j$. Thus, whereas $\lambda_j$ pertains to the antibiotic induced killing of taxon $j$, $K_j/r_j$ pertains to its growth, and both affect the change in its abundance. Third, $c_{j1}/\det(A)$ is the interaction quotient analogous to $c_{jk}/c_{jj}$ in Theorem \ref{main-eco}. It captures the community ecology mediated by the interaction network. The  difference from its analogue in Theorem \ref{main-eco} is due to the fact that in Theorem \ref{main-eco} the change in taxon $j$ was set independently, whereas here it is affected by the community. Lastly, the minus sign is due to the negative effect of the antibiotic on taxon $j$. This is then generalised from antibiotic effects on a single taxon $j$ to any set of antibiotic effects on any set of taxa (except the focal resistant taxon $1$) and to the tracking of any set of resistant taxa.

\newpage

\begin{example}[Mixed interactions]\label{Example1}
Here we walk through the calculations for the numerical example in Figure \ref{fig:Dynamics} panels g to i, with a mix of inhibitory ($\alpha_{ij}>0$) and facilitatory ($\alpha_{ij}<0$) interactions, for which we consider three antibiotic scenarios (no antibiotic, one target, two targets). There are three taxa, $x_1$, $x_2$ and $x_3$, where $x_1$ is the focal resistant taxon. The interaction matrix is:


\begin{align*}
A =  \begin{pmatrix}
1 &  3/4 & 1/2 \\
-1/4 & 1 & 1/2 \\
-1/4 & 3/8 & 1
\end{pmatrix}.
\end{align*}
To simplify computations, we let $K_i = r_i = 1$ for all $i =1,2,3$. 

Scenario i) In the absence of antibiotics, the system is at equilibrium. See Figure \ref{fig:Dynamics} panel (g). The relevant corresponding determinant and cofactors for $A$ are:


\begin{equation}
c_{21} = \frac{15}{16}, \quad c_{31} = -\frac{7}{8}, \quad
\det(A)  = \frac{51}{64}. 
\end{equation}

Scenario ii) Now consider an antibiotic that affects only taxon $x_2$, i.e. we increase $l_2$ from zero to some number $\la_2>0$, in this case $\la_2 =0.3$, see Figure \ref{fig:Dynamics} panel (h). The change in taxon $x_1$ is, according to formula (\ref{delta-1-a}),


\[
\de_1 = - \la_2 \frac{c_{21}}{\det(A)} = - \la_2 \frac{20}{17} = -\frac{6}{17} \approx -0.35.
\]

\noindent Hence, the focal resistant taxon decreases in abundance. The taxon $x_2$ itself also decreases (of course), but taxon $x_3$ increases. The system quickly reaches a new equilibrium. 

Scenario iii) Consider next an antibiotic with a stronger/broader spectrum. We represent this by adding an effect on $x_3$ in addition to the effect on $x_2$, i.e. we change $l_3$ from zero to some number $\la_3 > 0$, in this case $\la_3 = 0.5$,  see Figure \ref{fig:Dynamics} panel (i). Using the formula (\ref{delta-1-b}), we then get 


\[
\de_1 = - \la_2 \frac{c_{21}}{\det(A)} - \la_3 \frac{c_{31}}{\det(A)} = -  \la_2 \frac{20}{17} + \la_3 \frac{56}{51} = \frac{10}{51} \approx 0.2.
\]
In this case, making the spectrum stronger/broader not only changes the magnitude but also the direction of the effect on the focal resistant taxon. This scenario illustrates the complexities involved in predicting antibiotic effects in even a simple (three taxa) theoretical community.
\end{example}

\subsection*{Probabilistic results} \label{sec:prob}
Here we provide probabilistic results that do not require full information on parameter values. Again, we prioritize mathematical rigour, and an interpretation is given in the following section.

We suppose (see Assumptions~\ref{sec:assumptions}) that the interaction coefficients $\al_{ij}$ are all i.i.d. and let $\al_{ii} = 1$, for all $i = 1, \ldots, n$. Moreover, from the assumptions, we have $|\al_{ij}| < 1$, for all $i \neq j$. Recall that we assume that there exists a stable coexistence equilibrium. By \ref{Cor-1}, Appendix \ref{app:math_analysis}, this implies that $\det(A) > 0$. Hence, we will restrict to that set of matrices. Since $c_{jj}$ is a determinant of the $(n-1) \times (n-1)$-matrix obtained by deleting row and column $j$ from $A$, when $n$ is large, $c_{jj}$ and $\det(A)$ will have almost the same expectation.

\begin{Thm} \label{main-prob}
    Fix some $\ka > 0$. Suppose that all $\al_{ij}$ are i.i.d:s and are drawn from a probability distribution on $[-1,1]$ restricted to the set where $\det(A) \geq \ka > 0$. Let $\E(\al_{ij}) = \al$. There exists a constant $C > 0$ (which depends only on $n$ and $\ka$) such that if the standard deviation $\si(\al_{ij})$ of $\al$ satisfies $\si(\al) \leq C \al$ and $|\al| \leq C$, then, if the spectrum is made stronger (i.e. $\la_j \geq 0$ for all $j$ and at least one $\la_j > 0$),
    \begin{enumerate}[label=\roman*)]
        \item if $0 < \al < 1$, then $\E\left(\dfrac{c_{j1}}{\det(A)}\right) < 0$, and $\E(\de_1) > 0$ \label{case1}
        \item if $-1 < \al < 0$, then $\E\left(\dfrac{c_{j1}}{\det(A)} \right) > 0$, and $\E(\de_1) < 0$ \label{case2}
       \item if $\al = 0$, then $\E\left(\dfrac{c_{j1}}{\det(A)}\right) = 0$, and $\E(\de_1) = 0$. \label{case3}
    \end{enumerate}
\end{Thm}

\begin{Rem}
The reason to require $\det(A) \geq \ka$ is the singular behaviour of the quotient $c_{j1}/\det(A)$  when the determinant becomes small. 
\end{Rem}

From Lemma \ref{main-1} in Appendix \ref{app:math_analysis} (second part of the lemma; recall that $c_{jj}$ is just the determinant of an $(n-1) \times (n-1)$ principal matrix obtained from $A$), it follows that the expectations of $c_{jj}$ and $\det(A)$ are, respectively,  
\[
\E(c_{jj}) = (1-\al)^{n-2} (1 + \al(n-2)) \quad \text{ and } \quad \E(\det(A)) = (1-\al)^{n-1} (1 + \al(n-1)). 
\]

We see that there is a {\em singular} value of $\al$, (apart from $\al = 1$), when this expectation becomes zero. We get the singular values 
\begin{equation} \label{al-s}
\al_s = -\frac{1}{n-1} \quad \text{ and } \quad \al_s' = -\frac{1}{n-2},
\end{equation}
for the cases where $\det(A)$ and $c_{jj}$ are in the denominator of the interaction quotient, respectively. 
We therefore expect very wild behaviour close to this singular value. To the left of this singular point, $\al < \al_s$, the dynamics in the gLV-system seem to blow up, i.e., at least one taxon diverges to infinity. This behaviour aligns with numerical observations reported in Table \ref{tab:statistics_attractors_normal} (for $\al = -0.05$); for larger systems, $n=32,64$, the expectation $\alpha=-0.05$ is less than $-1/(n-1)$. 
To the right of this singular point, $\alpha>\alpha_s$, the quotient behaves nicely, as illustrated in Figure~\ref{fig:statistics_cofactor}. In that figure we record statistical averages for $c_{j1}/\det{A}$, sampled over many realizations of the interaction matrix $A$ (given $\det(A) > 0$). The observed behaviour (to the right of $\alpha_s$) suggests that the dynamics are convergent. 


\begin{Rem}
  The above theorem is local, i.e. it holds only for sufficiently small $\al$ and $\si(\al_{ij})$. However, it could well be that the theorem also holds in greater generality. Indeed, numerical results in Figure \ref{fig:statistics_cofactor} strongly support the idea that the statements in Theorem \ref{main-prob} hold true in greater generality. We can replace the condition on $\si(\al_{ij})$ and $|\al|$ in the above theorem with a statement requiring that the standard deviations of $c_{j1}$ and $\det(A)$ (or $c_{jj}$) are sufficiently small (see the discussion after the proof of Theorem \ref{main-2} in Supplement \ref{app:math_analysis} for more details). We leave this for further investigation. 
\end{Rem}

Notably, the Theorem remains true if we replace $c_{jj}$ by $\det(A)$ in all three cases, with the additional assumption that $\det(A) \geq \ka$. 

Theorem \ref{main-prob} above applies to the effect of a given change in $x_j$, independent of an external force (like an antibiotic). If the change in $x_j$ is induced by an antibiotic effect (the loss term in gLV), we replace $c_{jj}$ by $\det(A)$, and take the expectation of the formula (\ref{delta-1-a}). We get the following result: 
A change in antibiotic strength $l_j$ by  $\la_j$ ($l_j \mapsto l_j + \la_j$)  leads to an expected change $\de_1$ in taxon $1$, 
\begin{equation}
\mathbb{E}(\de_1) = - \la_j \frac{K_j}{r_j} \;\mathbb{E}\left ( \frac{c_{j1}}{\det(A)}\right ).
\end{equation}
Thus, Theorem \ref{main-prob} and the numerical results, see Figure  \ref{fig:statistics_cofactor}, suggest that the expectation value of $\de_1$ is positive if $\al > 0$. 

We may also predict the effect of varying $l_j$ for multiple taxa, $j\neq 1$. Again, let us assume that we vary $l_j$ for $j\in I$, by an amount $\la_j$. To  generalize the formula to target multiple taxa, by linear superposition, by (\ref{delta-1-b}), the summed effect is given by
\[
\mathbb{E}(\de_1) = - \sum_{j \in I} \la_j \frac{K_j }{r_j} \;\mathbb{E}\left ( \frac{c_{j1}}{\det(A)}\right ).
\]
For initial treatment (i.e., starting from the spectrum $\oli{l}=(0,0,\ldots,0)$), given the conditions in Theorem \ref{main-prob}, the gLV-model thus implies the following. If $\alpha>0$, a weaker ecological spectrum results in a lower expected abundance of the resistant taxon, and conversely, if $\al < 0$, a stronger spectrum results in a lower expected abundance.

As in the deterministic case, the formula can be used to track any set of resistant taxa.

\bigskip

\begin{table}[htp!] 
\centering
$(\alpha,\sigma)=(-0.05,0.05)$\\
\medskip
\begin{tabular}{c|cccc}
\hline\hline
 \text{n} & \text{equilibria (convergent)} & \text{oscillatory} & \text{divergent} & \text{coexistence} \\\hline
 4 & 10000 & 0 & 0 & 10000 \\
 8 & 10000 & 0 & 0 & 10000 \\
 16 & 10000 & 0 & 0 & 10000 \\
 32 & 0 & 0 & 10000 & 0 \\
 64 & 0 & 0 & 10000 & 0 \\
\hline\hline
\end{tabular}
\medskip \\
$(\alpha,\sigma)=(0,0.05)$\\
\medskip
\begin{tabular}{c|cccc}
\hline\hline
 \text{n} & \text{equilibria (convergent)} & \text{oscillatory} & \text{divergent} & \text{coexistence} \\\hline
 4 & 10000 & 0 & 0 & 10000 \\
 8 & 10000 & 0 & 0 & 10000 \\
 16 & 10000 & 0 & 0 & 10000 \\
 32 & 10000 & 0 & 0 & 9900 \\
 64 & 10000 & 0 & 0 & 5793 \\
\hline\hline
\end{tabular}
\medskip \\
$(\alpha,\sigma)=(0.05,0.05)$\\
\medskip
\begin{tabular}{c|cccc}
\hline\hline
 \text{n} & \text{equilibria (convergent)} & \text{oscillatory} & \text{divergent} & \text{coexistence} \\
 4 & 10000 & 0 & 0 & 10000 \\
 8 & 10000 & 0 & 0 & 10000 \\
 16 & 10000 & 0 & 0 & 10000 \\
 32 & 10000 & 0 & 0 & 9917 \\
 64 & 10000 & 0 & 0 & 6299 \\
\hline\hline
\end{tabular}
\caption{\label{tab:statistics_attractors_normal}Statistics for attractors in systems of dimension $n$ with interactions $\alpha_{ij}$ sampled from the truncated normal distribution with mean $\alpha$ and standard deviation $\sigma$ listed above the table. The table shows for $10'000$ iterations the number of runs that show convergence to an equilibrium, sustained oscillations, or divergence to infinity, respectively. It also shows the number of equilibria that were coexistence equilibria, i.e. all $x_i > 0$. For details see Appendix~\ref{app:num_algorithm}.
}
\end{table}

\begin{figure}[htp!]
    \begin{overpic}[width=\textwidth]{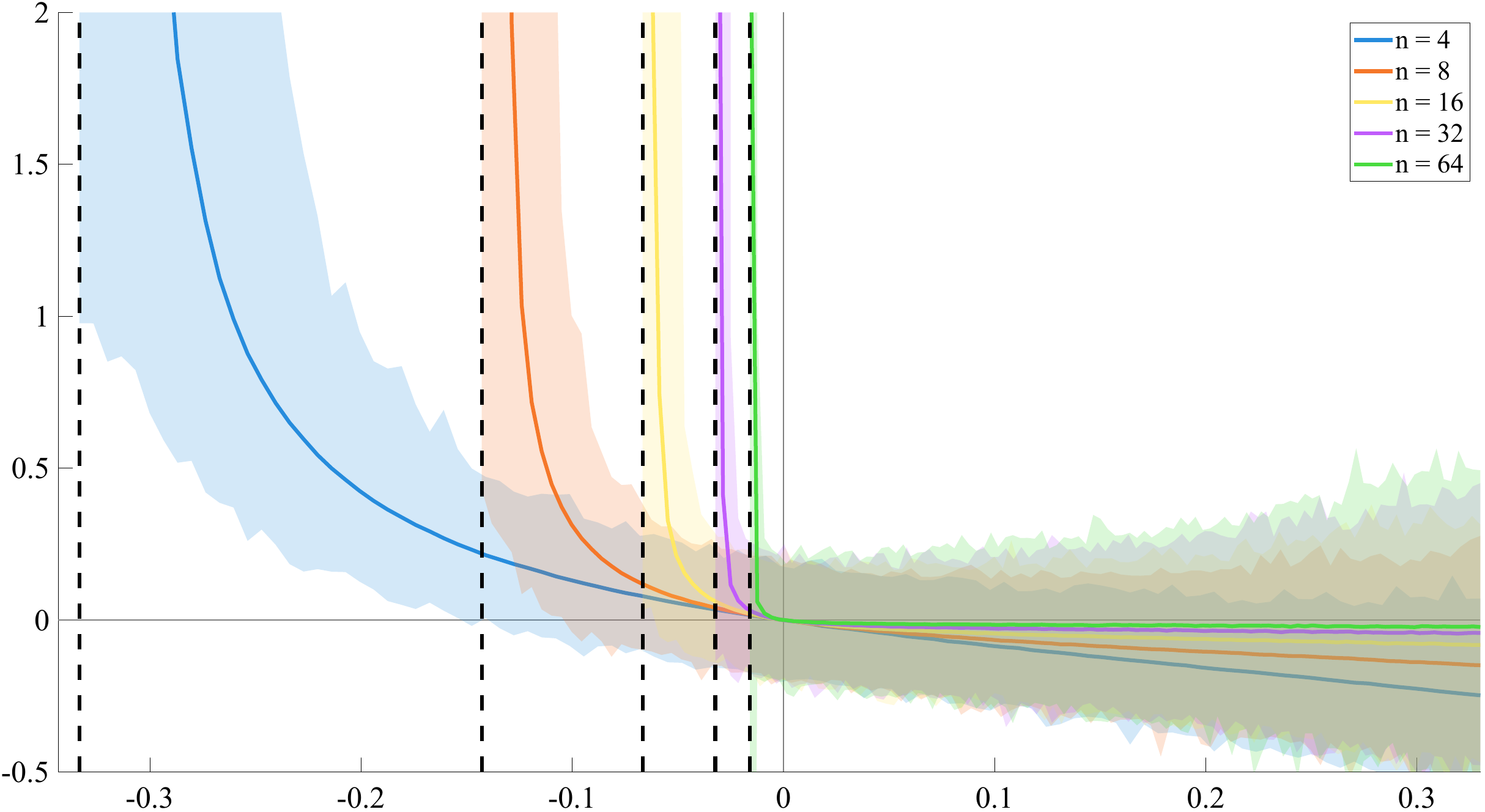} 
    \put(35,-3){Mean interaction coefficient $\alpha$}
    \put(4,55){\rotatebox{90}{\footnotesize $\alpha_s(4)$}}
    \put(31,55){\rotatebox{90}{\footnotesize $\alpha_s(8)$}}
    \put(42,55){\rotatebox{90}{\footnotesize $\alpha_s(16)$}}
    \put(47,55){\rotatebox{90}{\footnotesize $\alpha_s(32)$}}
    \put(49.5,55){\rotatebox{90}{\footnotesize $\alpha_s(64)$}}
    \put(-5,15){\rotatebox{90}{$ \text{Interaction quotient } \langle c_{j1}/\det (A) \rangle $}}
    \end{overpic}
    \smallskip
    \caption{\label{fig:statistics_cofactor} 
    Statistics for the interaction quotient $c_{j1}/\det (A)$ based on numerical simulations (for details  see Appendix~\ref{app:num_algorithm}). 
    Interaction matrices with $\det (A)>0$ were randomly sampled from a truncated normal distribution $\mathcal{N}_{[-1,1]}(\alpha,\sigma)$ with mean $\alpha$ and standard deviation $\sigma=0.05$ for an ensemble over $N_\text{re}=10'000$ realizations. The ensemble average  $\langle c_{j1}/\det (A)\rangle$ was recorded for varying system sizes while varying the mean value $\alpha$. Shaded regions of the same colour indicate minimal and maximal values for the ensemble.
    The average of the interaction quotient is negative for positive $\al$ and vice versa. 
    While the range of $\alpha$ is unbounded for positive values, it is bounded to the left for negative values, indicated by vertical dashed lines located at $\alpha=\alpha_s$,  see \eqref{al-s} in the section of probabilistic results. As $\al$ approaches this singular value, the interaction quotient diverges. }
\end{figure}

\subsection*{Interpretation of the probabilistic results}
Whilst the deterministic results above are key to understanding the effect of antibiotic spectrum in the gLV, the fact that we do not have complete information on parameter values --- in particular, interaction coefficients --- makes these results hard to apply to real world ecological or medical questions. We therefore also provide a probabilistic version of those results based on expected values (averages) of interaction coefficients.

The probabilistic results in the previous section are analogous to the deterministic ones discussed earlier, and they have analogous interpretations. They are less general than the deterministic theorems, since they were proven only for small $\alpha = \E(\alpha_{ij})$ and $\sigma(\alpha_{ij})$, and supported by numerical simulations shown in Figure \ref{fig:statistics_cofactor}. The key difference is that the change in abundance of taxon $1$, as well as  interaction coefficients and the interaction quotient, are expressed as expected values. The expected value of the interaction quotient depends on the expected value of the interaction coefficients, i.e., $\alpha = \E(\alpha_{ij})$. The relationship between the two is shown in Figure \ref{fig:statistics_cofactor}, which plots the \textit{interaction quotient} as a function of the average of interaction coefficients for numerical simulations. As discussed above, the interaction quotient scales the effect of targeting a taxon $j$ with an antibiotic on the abundance of a focal resistant taxon $1$.

Qualitatively, we may identify three scenarios, captured by Theorem \ref{main-prob}. Recall that gLV interaction coefficients are positive for inhibitory and negative for facilitatory interactions.

In the first scenario, interactions are dominated by inhibition ($\alpha>0$), and the expected value of the interaction quotient is negative (Theorem \ref{main-prob} case \textit{i} and Figure \ref{fig:statistics_cofactor}). If an arbitrary taxon $j$ is targetted with an antibiotic, the expected value of the change in the focal resistant taxon $1$ ($\E(\delta_1)$) is positive. That is, the resistant taxon $1$ increases in abundance. This is competitive release. As an extension of this, initiating treatment with an antibiotic with a stronger ecological spectrum is expected to induce a greater increase in the focal resistant taxon as compared to an antibiotic with a weaker spectrum. Likewise, if a microbial community is at equilibrium under antibiotic exposure, and the antibiotic is switched to a stronger spectrum, the focal resistant taxon is expected to increase. Conversely, when switching from a stronger to a weaker spectrum (as in de-escalation), the focal resistant taxon is expected to decrease. Figure \ref{fig:statistics_cofactor} indicates that, for $\alpha>0$, the magnitude of the interaction quotient is smaller for larger community sizes $n$. This is because the taxon $j$ targeted by the antibiotic is a smaller portion of the community for large communities. Figure \ref{fig:Dynamics} row $1$ illustrates community dynamics for a numerical example consistent with this scenario. An antibiotic induces competitive release (b) and a stronger/broader spectrum antibiotic induces greater competitive release (c).

In the second scenario, interactions are dominated by facilitation ($\alpha<0$), and the expected value of the interaction quotient is positive (Theorem \ref{main-prob} case \textit{ii} and Figure \ref{fig:statistics_cofactor}). Qualitatively, the expected effects on the focal resistant taxon $1$ are now opposite to those in the first scenario. Targeting taxon $j$ is expected to decrease the resistant taxon $1$, and a stronger ecological spectrum is expected to yield a larger decrease compared to a weaker spectrum. However, as seen in Figure \ref{fig:statistics_cofactor}, as the average of interaction coefficients ($\alpha$) decreases, the expected value of the interaction quotient (for large $N$, $c_{j1}/\det(A)$ and $c_{j1}/c_{jj}$ are very similar) grows towards infinity. This indicates a limit to the applicability of the gLV model (with or without antibiotic). This limit is given by $\alpha = -1/(n-1), i \neq j$ (in the case $c_{j1}/\det(A)$), and its biological meaning is discussed below. For large $n$ this limit is close to $0$, as illustrated in Figure \ref{fig:statistics_cofactor}. Figure \ref{fig:Dynamics} row $2$ illustrates community dynamics for a numerical example consistent with this scenario. As in the previous scenario, a stronger/broader spectrum induces a greater change in the focal resistant taxon.

In the third scenario, there is no dominance of either inhibitory or facilitatory interactions, but interactions are on average neutral ($\alpha=0$), and the expected value of the interaction quotient is $0$ (Theorem \ref{main-prob} case \textit{iii)} and Figure \ref{fig:statistics_cofactor}). Here, the expected value of the change in the focal resistant taxon ($\E(\delta_1)$) is $0$ (not illustrated in Figure \ref{fig:Dynamics}).

Lastly, we stress that these scenarios represent probabilistic expectations only.

\section*{Discussion}
The mathematical analysis (above and \hyperref[app:appendix]{Appendix}) is abstract, but it has a deeper interpretation that is intuitively helpful in understanding the underlying community ecology.

Consider again the small microbial community in Figure \ref{Graph}. If \textit{blue taxon 2} is targeted with an antibiotic, what is the net effect on \textit{the red resistant taxon 1}? Clearly, \textit{blue 2} will decrease in abundance under antibiotic pressure, and this will change the strengths of its effects on \textit{the red taxon}, but what are those effects?

In a community with many taxa, there is a network of interactions. A given taxon affects the focal resistant taxon via non-zero interactions. These interactions form paths across the network that start at that taxon and end at the focal one. 

There may be several paths connecting two taxa in a network. In the network in Figure \ref{Graph} there is a path from \textit{blue 2} directly to \textit{the red taxon}, representing a direct effect. There is also a path from \textit{blue 2} through \textit{blue 3} to \textit{the red taxon}, representing an indirect effect. In addition, there is a closed path (loop), where \textit{blue 2} affects \textit{blue 3} that in turn feeds back on \textit{blue 2}, thus modulating its abundance and hence its effects. Then there is feedback from the \textit{red taxon} itself, yielding additional closed paths.


The key measure in our analysis --- the interaction quotient $c_{j1}/\det(A)$ --- takes the product of interaction coefficients along each path starting at taxon $j$ and ending at the focal taxon $1$, and sums these products over all paths weighted by the feedback loops that each path is affected by. Thus, it is a summary --- specifically, a weighted sum --- of all effects of taxon $j$ on taxon $1$ throughout the entire interaction network. In Appendix \ref{app:math_analysis}, Section \ref{combinatorics}, we explicate this interpretation and its basis in the combinatorial interpretation of determinants.

Like all scientific representations of reality, the generalised Lotka-Volterra model is a simplification, and at some point it breaks down. In particular, it is well known that under some conditions, it supports unbounded growth. Here we give an interpretation in terms of the network of interactions.

Consider an interaction network with $n$ taxa. Each taxon inhibits itself with strength $1$ ($\al_{ii}=1$). It also has one incoming effect from each of the $n-1$ other taxa. If these average $-1/(n-1)$, they sum to $-1$, and the sum total of all incoming effects (self-interaction and effects of other taxa) is $0$. An analogous argument applies to outgoing effects. Thus, the singular value identified above $\alpha = \E(\al_{ij}) = -1/(n-1), i \neq j$ marks the boundary between net inhibitory and net facilitatory regimes. In the latter regime, a key simplification in the gLV model becomes salient: with no explicit representation of external resources to limit the scope for facilitation, taxa can support each other's growth without bound, and the interaction quotient goes to infinity.


In conclusion, we have developed formal theory to bridge the gaps between microbial community ecology, the evolution of resistance, and antibiotic stewardship policy. From the model and assumptions, we have derived deterministic results that use full information on ecological parameters, and provided probabilistic results based on statistical averages. Though the mathematics of the generalised Lotka-Volterra model is complex, for the matters investigated here, it boils down to surprisingly simple formulas with direct biological interpretations.

In our analysis, the ecological interactions between taxa are a key determinant of the outcome. The measurement of interactions in microbial communities is a work in progress, but a recent assessment found interactions were mostly inhibitory \cite{Palm_Fost}. 
Given those results, our findings support the current stewardship policy of favouring antibiotics with a narrow (or weak) spectrum.

Though clearly preliminary, this conclusion illustrates the long term goal, towards which this study aims to contribute: formally rigorous, empirically validated and clinically actionable theory for the evolutionary ecology of antibiotic resistance.


\section*{Acknowledgments}
We thank Sam Brown, Jessica Abbott, Ben Metcalf, Maria Martignoni and Magnus Wiktorsson for helpful discussions and comments. This work was supported by grants from the Cystic Fibrosis Foundation, the Swedish Research Council (ref. no. $2024-04198$) and The Royal Swedish Academy of Sciences (Magnussons Fond). 

\printbibliography

@article{Nakajima-1992,
title = {Sensitivity and stability of flow networks},
journal = {Ecological Modelling},
volume = {62},
number = {1},
pages = {123-133},
year = {1992},
issn = {0304-3800},
doi = {https://doi.org/10.1016/0304-3800(92)90085-S},
url = {https://www.sciencedirect.com/science/article/pii/030438009290085S},
author = {Hisao Nakajima},
}

@article{Novak-etal-2026,
   author = "Novak, Mark and Yeakel, Justin D. and Noble, Andrew E. and Doak, Daniel F. and Emmerson, Mark and Estes, James A. and Jacob, Ute and Tinker, M. Timothy and Wootton, J. Timothy",
   title = "Characterizing Species Interactions to Understand Press Perturbations: What Is the Community Matrix?", 
   journal= "Annual Review of Ecology, Evolution, and Systematics",
   year = "2016",
   volume = "47",
   number = "Volume 47, 2016",
   pages = "409-432",
   doi = "https://doi.org/10.1146/annurev-ecolsys-032416-010215",
   url = "https://www.annualreviews.org/content/journals/10.1146/annurev-ecolsys-032416-010215",
   publisher = "Annual Reviews",
   issn = "1545-2069",
   type = "Journal Article",
  }

@article {Lechon-Alonso-etal,
    AUTHOR = {Lechón-Alonso, Pablo and Kundu, Srilena and Lemos-Costa, Paula and Capitán, José A, and Allesina, Stefano},
     TITLE = {Robust coexistence in competitive ecological communities},
   JOURNAL = {Nat. Commun.},
  FJOURNAL = {Nature Communications},
    VOLUME = {17},
      YEAR = {2026},
    NUMBER = {2637},
     PAGES = {},
      ISSN = {},
   MRCLASS = {},
  MRNUMBER = {},
MRREVIEWER = {},
       DOI = {10.1038/s41467-026-69151-3},
       URL = {https://doi.org/10.1038/s41467-026-69151-3},
}

@article {Clenet-Massol-Najim,
    AUTHOR = {Clenet, Maxime and Massol, Fran\c cois and Najim, Jamal},
     TITLE = {Equilibrium and surviving species in a large
              {L}otka-{V}olterra system of differential equations},
   JOURNAL = {J. Math. Biol.},
  FJOURNAL = {Journal of Mathematical Biology},
    VOLUME = {87},
      YEAR = {2023},
    NUMBER = {1},
     PAGES = {Paper No. 13, 32},
      ISSN = {0303-6812,1432-1416},
   MRCLASS = {92D40 (60B20 60G70)},
  MRNUMBER = {4606239},
       DOI = {10.1007/s00285-023-01939-z},
       URL = {https://doi.org/10.1007/s00285-023-01939-z},
}

@article {Clenet-El-Najim,
    AUTHOR = {Clenet, Maxime and El Ferchichi, Hafedh and Najim, Jamal},
     TITLE = {Equilibrium in a large {L}otka-{V}olterra system with pairwise
              correlated interactions},
   JOURNAL = {Stochastic Process. Appl.},
  FJOURNAL = {Stochastic Processes and their Applications},
    VOLUME = {153},
      YEAR = {2022},
     PAGES = {423--444},
      ISSN = {0304-4149,1879-209X},
   MRCLASS = {34F05 (15B52 60B20 60G70 92D25)},
  MRNUMBER = {4486707},
MRREVIEWER = {Javier\ L\'opez-de-la-Cruz},
       DOI = {10.1016/j.spa.2022.08.004},
       URL = {https://doi.org/10.1016/j.spa.2022.08.004},
}

@article {Bizeul-Najim,
    AUTHOR = {Bizeul, Pierre and Najim, Jamal},
     TITLE = {Positive solutions for large random linear systems},
   JOURNAL = {Proc. Amer. Math. Soc.},
  FJOURNAL = {Proceedings of the American Mathematical Society},
    VOLUME = {149},
      YEAR = {2021},
    NUMBER = {6},
     PAGES = {2333--2348},
      ISSN = {0002-9939,1088-6826},
   MRCLASS = {15B52 (60B20 60G70 92D40)},
  MRNUMBER = {4246786},
MRREVIEWER = {Vladislav\ Kargin},
       DOI = {10.1090/proc/15383},
       URL = {https://doi.org/10.1090/proc/15383},
}

@article {Ad-Ta-To-1,
    AUTHOR = {Takeuchi, Yasuhiro and Adachi, Norihiko and Tokumaru,
              Hidekatsu},
     TITLE = {The stability of generalized {V}olterra equations},
   JOURNAL = {J. Math. Anal. Appl.},
  FJOURNAL = {Journal of Mathematical Analysis and Applications},
    VOLUME = {62},
      YEAR = {1978},
    NUMBER = {3},
     PAGES = {453--473},
      ISSN = {0022-247X},
   MRCLASS = {34D05 (92A05)},
  MRNUMBER = {477317},
MRREVIEWER = {A.\ F.\ Iz\'e},
       DOI = {10.1016/0022-247X(78)90139-7},
       URL = {https://doi.org/10.1016/0022-247X(78)90139-7},
}

@article {Ad-Ta-To-2,
    AUTHOR = {Takeuchi, Yasuhiro and Adachi, Norihiko and Tokumaru,
              Hidekatsu},
     TITLE = {Global stability of ecosystems of the generalized {V}olterra
              type},
   JOURNAL = {Math. Biosci.},
  FJOURNAL = {Mathematical Biosciences},
    VOLUME = {42},
      YEAR = {1978},
    NUMBER = {1-2},
     PAGES = {119--136},
      ISSN = {0025-5564,1879-3134},
   MRCLASS = {92A15},
  MRNUMBER = {529099},
       DOI = {10.1016/0025-5564(78)90010-X},
       URL = {https://doi.org/10.1016/0025-5564(78)90010-X},
}

@article {Ad-Ta-1,
    AUTHOR = {Takeuchi, Yasuhiro and Adachi, Norihiko},
     TITLE = {The existence of globally stable equilibria of ecosystems of
              the generalized {V}olterra type},
   JOURNAL = {J. Math. Biol.},
  FJOURNAL = {Journal of Mathematical Biology},
    VOLUME = {10},
      YEAR = {1980},
    NUMBER = {4},
     PAGES = {401--415},
      ISSN = {0303-6812,1432-1416},
   MRCLASS = {92A17 (34D05 90C33)},
  MRNUMBER = {602257},
MRREVIEWER = {J.\ M.\ Cushing},
       DOI = {10.1007/BF00276098},
       URL = {https://doi.org/10.1007/BF00276098},
}

@article {Ad-Ta-2,
    AUTHOR = {Takeuchi, Yasuhiro and Adachi, Norihiko},
     TITLE = {Existence of stable equilibrium point for dynamical systems of
              {V}olterra type},
   JOURNAL = {J. Math. Anal. Appl.},
  FJOURNAL = {Journal of Mathematical Analysis and Applications},
    VOLUME = {79},
      YEAR = {1981},
    NUMBER = {1},
     PAGES = {141--162},
      ISSN = {0022-247X},
   MRCLASS = {34D20 (90A30 92A15)},
  MRNUMBER = {603382},
MRREVIEWER = {A.\ F.\ Iz\'e},
       DOI = {10.1016/0022-247X(81)90015-9},
       URL = {https://doi.org/10.1016/0022-247X(81)90015-9},
}

@article {Ad-Ta-3,
    AUTHOR = {Takeuchi, Yasuhiro and Adachi, Norihiko},
     TITLE = {Stable equilibrium of systems of generalized {V}olterra type},
   JOURNAL = {J. Math. Anal. Appl.},
  FJOURNAL = {Journal of Mathematical Analysis and Applications},
    VOLUME = {88},
      YEAR = {1982},
    NUMBER = {1},
     PAGES = {157--169},
      ISSN = {0022-247X},
   MRCLASS = {34D20 (92A15 92A17)},
  MRNUMBER = {661409},
MRREVIEWER = {Harlan\ W.\ Stech},
       DOI = {10.1016/0022-247X(82)90183-4},
       URL = {https://doi.org/10.1016/0022-247X(82)90183-4},
}

@article{Goh,
 ISSN = {00030147, 15375323},
 URL = {http://www.jstor.org/stable/2459985},
 author = {B. S. Goh},
 journal = {The American Naturalist},
 number = {977},
 pages = {135--143},
 publisher = {[University of Chicago Press, American Society of Naturalists]},
 title = {Global Stability in Many-Species Systems},
 urldate = {2025-12-15},
 volume = {111},
 year = {1977}
}

@website{WHO_AMR,
    author = {World Health Organisation},
    year = {2023},
    title = {Antimicrobial resistance},
    subtitle = {Fact Sheet},
    url = {https://www.who.int/news-room/fact-sheets/detail/antimicrobial-resistance},
    note = {Accessed = 2025-02-18} 
}

@article{ESCMID_Stew,
    author ={Dyar, O.J. AND Beović, B. AND Pulcini, C. AND Tacconelli, E. AND  Hulscher, M. AND Cookson, B. AND Ashiru-Oredope, D. AND Barcs, I. AND Blix, H.S. AND Buyle, F. AND Chowers, M. Čižman, M. AND Cookson, B. AND Del Pozo, J.L. AND Deptula, A. AND Dumpis, U. AND Florea, D. AND van de Garde, E. AND Geffen, Y. AND Giske, C.G. AND Grau, S. AND Hajdú, E. AND Hell, M. AND Hondo, Ł. AND Hussein, K. AND Huttner, B. AND Kern, W. AND Kernéis, S. AND Knepper, V. AND Kofteridis, D. AND Kostyanev, T. AND Kuijper, E. AND Lebanova, H. AND Lewis, R. AND Cordina, C.M. AND Matulionyte, R. AND Maurer, F. AND Messiaen, P. AND Miciuleviciene, J. AND Mrhar, A. AND Nabuurs-Franssen, M. AND Naesens, R. AND Oxacelay, C. AND Pagani, L. AND Paño-Pardo, J.R. AND Paul, M. AND Petrikkos, G. AND Pluess-Suard, C. AND Popescu, G.A. AND Porsche, U. AND Prins, J. AND Pulcini, C. AND Rello, J. AND Rodríguez-Baño, J. AND Rossolini, G.M. AND Salzberger, B. AND Seme, K. AND Simonsen, G.S. AND Sînziana, M. AND  Skovgaard, S. AND Smith, I. AND Sönsken, U. AND Soriano, A. AND Sviestiņa, I. AND Szilagyi, E. AND Tängdén, T. AND Tattevin, P. AND Tsioutis, C. AND Vilde, A. AND Wanke-Rytt, M. AND Wechsler-Fördös, A. AND Zarb, P.},
    title = {ESCMID generic competencies in antimicrobial prescribing and stewardship: towards a European consensus},
    journal = {Clinical Microbiology and Infection},
    volume = 25,
    number = 1,
    pages = 13-19,
    year = 2019,
    DOI = {https://doi.org/10.1016/j.cmi.2018.09.022},
    }

@report{WHO_Stew,
author = {World Health Organization},
title = {Antimicrobial stewardship programmes in health-care facilities in low- and
middle-income countries. A practical toolkit},
isbn = {978-92-4-151548-1},
year = {2019}
}

@article{LID_Popgen,
 title = {Population biological principles of drug-resistance evolution in infectious diseases},
author = {zur Wiesch, Pia Abel and Kouyos, Roger and Engelstädter, Jan and Regoes, Roland R and Bonhoeffer, Sebastian},
year = {2011},
doi ={doi: 10.1016/S1473-3099(10)70264-4},
journal = {The Lancet Infectious Diseases},
pages = {236- 247},
volume = {11},
issue = {3},
url = {https://doi.org/10.1016/S1473-3099(10)70264-4}
}

@report{EU_Prudent,
title = {EU Guidelines for the prudent use of antimicrobials in human health},
author = {European Commission},
journal = {Official Journal of the European Union},
year = {2017},
issue = {C 212},
url = {https://eur-lex.europa.eu/legal-content/EN/ALL/?uri=CELEX:52017XC0701(01)}
}

@website{WHO_Spectrum,
author = {World Health Organization},
year = {2019},
url = {https://www.who.int/news/item/18-06-2019-in-the-face-of-slow-progress-who-offers-a-new-tool-and-sets-a-target-to-accelerate-action-against-antimicrobial-resistance},
title = {In the face of slow progress, WHO offers a new tool and sets a target to accelerate action against antimicrobial resistance},
note = {acessed 2025-02-21}
}

@website{EU_Amount,
title = {European Health Union: EU steps up the fight against antimicrobial resistance},
author = {European Commission},
year = {2023},
url = {https://ec.europa.eu/commission/presscorner/detail/en/ip_23_3187},
note = {accessed 2025-02-21}
}

@website{WHO_Choice,
title = {AWaRe classification of antibiotics for evaluation and monitoring of use, 2023},
author = {World Health Organization},
year = {2023},
url = {https://www.who.int/publications/i/item/WHO-MHP-HPS-EML-2023.04},
note = {accessed 2025-2-21}
}

@article{NatRevTarget,
author = {Theuretzbacher, Ursula and Blasco, Benjamin and Duffey, Maëlle and Piddock, Laura J. V.},
year = {2023},
title = {Unrealized targets in the discovery of antibiotics for Gram-negative bacterial infections},
journal = {Nature Reviews Drug Discovery},
pages = {957-975},
volume = {22},
issue = {12},
url = {https://doi.org/10.1038/s41573-023-00791-6}
}

@article{Conc_Sites,
title = {Bystander Selection for Antimicrobial Resistance: Implications for Patient Health},
author = {Morley, Valerie J. and Woods, Robert J. and Read, Andrew F.},
year ={2019},
journal = {Trends in Microbiology},
pages = {864-877},
volume = {27},
issue = {10},
url ={https://doi.org/10.1016/j.tim.2019.06.004}
}

@article{Tedijanto,
author = {Christine Tedijanto  and Scott W. Olesen  and Yonatan H. Grad  and Marc Lipsitch },
title = {Estimating the proportion of bystander selection for antibiotic resistance among potentially pathogenic bacterial flora},
journal = {Proceedings of the National Academy of Sciences},
volume = {115},
number = {51},
pages = {E11988-E11995},
year = {2018},
URL = {https://www.pnas.org/doi/abs/10.1073/pnas.1810840115}
}

@article{Palm_Fost,
author = {Jacob D. Palmer  and Kevin R. Foster },
title = {Bacterial species rarely work together},
journal = {Science},
volume = {376},
number = {6593},
pages = {581-582},
year = {2022},
URL = {https://www.science.org/doi/abs/10.1126/science.abn5093}
}

@article{Day_Read,
title = {Is selection relevant in the evolutionary emergence of drug resistance?},
author = {Day, Troy and Huijben, Silvie and Read, Andrew F.},
year = {2015},
journal = {Trends in Microbiology},
pages = {126-133},
volume = {23},
issue = {3},
url = {https://doi.org/10.1016/j.tim.2015.01.005}
}

@article{Coyte_Foster,
author = {Katharine Z. Coyte  and Jonas Schluter  and Kevin R. Foster },
title = {The ecology of the microbiome: Networks, competition, and stability},
journal = {Science},
volume = {350},
number = {6261},
pages = {663-666},
year = {2015},
doi = {10.1126/science.aad2602},
URL = {https://www.science.org/doi/abs/10.1126/science.aad2602}
}

@article{gLV_Rev,
title = {Microbial communities as dynamical systems},
journal = {Current Opinion in Microbiology},
volume = {44},
pages = {41-49},
year = {2018},
note = {Microbiota},
issn = {1369-5274},
doi = {https://doi.org/10.1016/j.mib.2018.07.004},
url = {https://www.sciencedirect.com/science/article/pii/S1369527418300092},
author = {Didier Gonze and Katharine Z Coyte and Leo Lahti and Karoline Faust}
}

@article{Davis_LV,
  title={Methods of quantifying interactions among populations using Lotka-Volterra models},
  author={Davis, Jacob D and Oliven{\c{c}}a, Daniel V and Brown, Sam P and Voit, Eberhard O},
  journal={Frontiers in Systems Biology},
  volume={2},
  pages={1021897},
  year={2022},
  publisher={Frontiers Media SA}
}

@article{Sundius,
  title={Defining the benefits of antibiotic resistance in commensals and the scope for resistance optimization},
  author={Wollein Waldetoft, Kristofer and Sundius, Sarah and Kuske, Rachel and Brown, Sam P},
  journal={Mbio},
  volume={14},
  number={1},
  pages={e01349--22},
  year={2023},
  publisher={American Society for Microbiology 1752 N St., NW, Washington, DC}
}

@article{Faith_stable,
author = {Jeremiah J. Faith  and Janaki L. Guruge  and Mark Charbonneau  and Sathish Subramanian  and Henning Seedorf  and Andrew L. Goodman  and Jose C. Clemente  and Rob Knight  and Andrew C. Heath  and Rudolph L. Leibel  and Michael Rosenbaum  and Jeffrey I. Gordon },
title = {The Long-Term Stability of the Human Gut Microbiota},
journal = {Science},
volume = {341},
number = {6141},
pages = {1237439},
year = {2013},
doi = {10.1126/science.1237439},
URL = {https://www.science.org/doi/abs/10.1126/science.1237439}}

@article{AspenbergMartensWaldetoft,
author = {Aspenberg, Magnus and Martens, Erik A. and  Waldetoft, Kristofer Wollein},
title = {Necessary condition for stable equilibria in the Lotka-Volterra equations},
journal = {arXiv preprint, arXiv:2512.13347},
volume = {},
number = {},
pages = {},
year = {2025},
doi = {},
URL = {}
}

@article{Allesina_100,
   author = "Allesina, Stefano",
   title = "One Hundred Years of Lotka–Volterra Models",
   journal = "Annual Review of Ecology, Evolution, and Systematics",
   issn = "1543-592X",
   year = "2026",
   publisher = "Annual Reviews",
   url = "https://www.annualreviews.org/content/journals/10.1146/annurev-ecolsys-102924-040829",
   doi = "https://doi.org/10.1146/annurev-ecolsys-102924-040829"
}

\appendix
\section*{Appendix\label{app:appendix}}
\section{Numerical frequency sampling of attractors\label{app:num_algorithm}}


To obtain statistics for the frequency of different  (stable) attractors that occur in the dynamical system defined by equations~\eqref{eq:model}, we numerically integrated Equations~\eqref{eq:model} and analyzed different asymptotic dynamic behaviours. 
To obtain data for different asymptotic dynamic behaviours, we used the following procedure:
\begin{enumerate}
    
    \item {\it Realization of initial conditions.} For each realization, we computed the trajectories using initial conditions sampled i.i.d. uniformly from the positive domain $x_i(0)\in[0,L] $ for all $i=1,\ldots, n$. Note that trajectories with such initial conditions remain forever positive. We sampled $n_i$ initial conditions for $n_r$ realization of the interaction matrices.

    \item {\it Realizations of interaction matrix.}
    This process was repeated for $N_\text{re}$ realizations of the interaction matrix $A$.
    Matrix components were randomly drawn (i.i.d.) with $\alpha_{ij}\in[-1,1]$, using either the  
    \begin{enumerate}[label=\roman*)]
        \item \textit{uniform distribution} $\mathcal{U}_{[-1,1]}$; or the
        \item \textit{truncated normal distribution} $\mathcal{N}_{[-1,1]}$ with mean $\alpha$ and standard deviation $\sigma$.
    \end{enumerate}
    
    \item  {\it Classification of asymptotic dynamic behaviour.} 
    We computed trajectories for a time of $T$ units. Transient behaviour was assumed to subside after $T_t$ time units, and for $t\in \mathcal{T}:=[0.5T,T]$ we measured the 'amplitude' of the signal, defined by $\mathcal A:=\max_{t\in\mathcal{T}}||x(t)||- \min_{t\in \mathcal{T}}||x(t)||$ where $||x(t)||:=(\sum_{j=1}^n x_j(t)^2)^{1/2}$ is the Euclidean norm of the trajectory at time $t$. 
    We classified the asymptotic  behaviour of each trajectory into one of three distinct dynamic behaviours: 
    \begin{enumerate}[label=\roman*)]
        \item \textit{stable equilibrium}, if $\mathcal A<\epsilon_c$; 
        \item \textit{oscillatory motion}, if $\epsilon_c\leq\mathcal  A \leq\epsilon_d$ \\ (i.e., trajectories bounded with non-zero amplitude: periodic, quasiperiodic, or chaotic); 
        \item \textit{divergent dynamics}, if $\mathcal A>\epsilon_d$.
    \end{enumerate}
    Stable equilibria were further classified into \textit{(partial-existence equilibria} of order $k\geq 0$. The presence of multiple stable equilibria was tested by comparing $||x(t)||_2$ at time $T$, rounded to two decimals, over all sampled initial conditions for each realization. 
\end{enumerate}
Numerical integration of \eqref{eq:model} was carried out in MATLAB using a stiff solver (\verb|ode23s|). The sampling parameters were $N_\text{re}=100$, $N_\text{ic} = 100$ and $L=5$, and the model parameters were $\epsilon_c = 10^{-2}, \epsilon_d = 10^3, T=200$, $T_t = 0.5 T = 100$, and $r_i=K_i=1$ with $l_i=0$. 
Results are shown in Tables~\ref{tab:statistics_attractors_uniform}  and \ref{tab:statistics_attractors_normal}, respectively. Further statistical results for the occurrence of various possible partial-existence equilibria are shown in Figure \ref{fig:statistics_SubCanonical_uniform}. 

\begin{table}[htp!]
\centering
\begin{tabular}{c|cccc}
\hline\hline
 $n$ & \text{Equilibria (convergent)} & \text{Oscillatory} & \text{Divergent} & \text{coexistence} \\ \hline
 4 & 9309 & 91 & 600 & 2611 \\
 8 & 6852 & 735 & 2413 & 100 \\
 16 & 1906 & 531 & 7563 & 0 \\
 32 & 0 & 0 & 10000 & 0 \\
 64 & 0 & 0 & 10000 & 0 \\
\hline\hline
\end{tabular}
\caption{\label{tab:statistics_attractors_uniform}Statistics for attractors in systems of dimension $n$ with interactions sampled from the uniform distribution $\alpha_{ij}\sim\mathcal{U}_{[-1,1]}$. 
}
\end{table}

These numerical procedures for the uniform distribution of $\alpha_{ij}\sim\mathcal{U}_{[-1,1]}$ (Table~\ref{tab:statistics_attractors_uniform}) show that most trajectories tend to equilibria or diverge in the asymptotic time limit for larger $n$.  Some trajectories display oscillatory behaviour, but their occurrence is  rare in our sampling (note that 'oscillatory behaviour' simply refers to any behaviour not converging to an equilibrium, see also Appendix \ref{app:num_algorithm}). As the number of species $n$ increases, divergent behaviour becomes very prominent.  For a discussion of results with the normal distribution, see main text.
It is biologically more natural to assume that interaction coefficients (excluding self-interactions) are non-uniformly distributed. Unfortunately, at present, experimental data with detailed network interactions are to our knowledge unavailable. Therefore, we assumed that interaction coefficients are normally distributed, as discussed in the main text. 

\medskip

\begin{figure}[htp!]
\centering
\begin{overpic}[width=\textwidth]{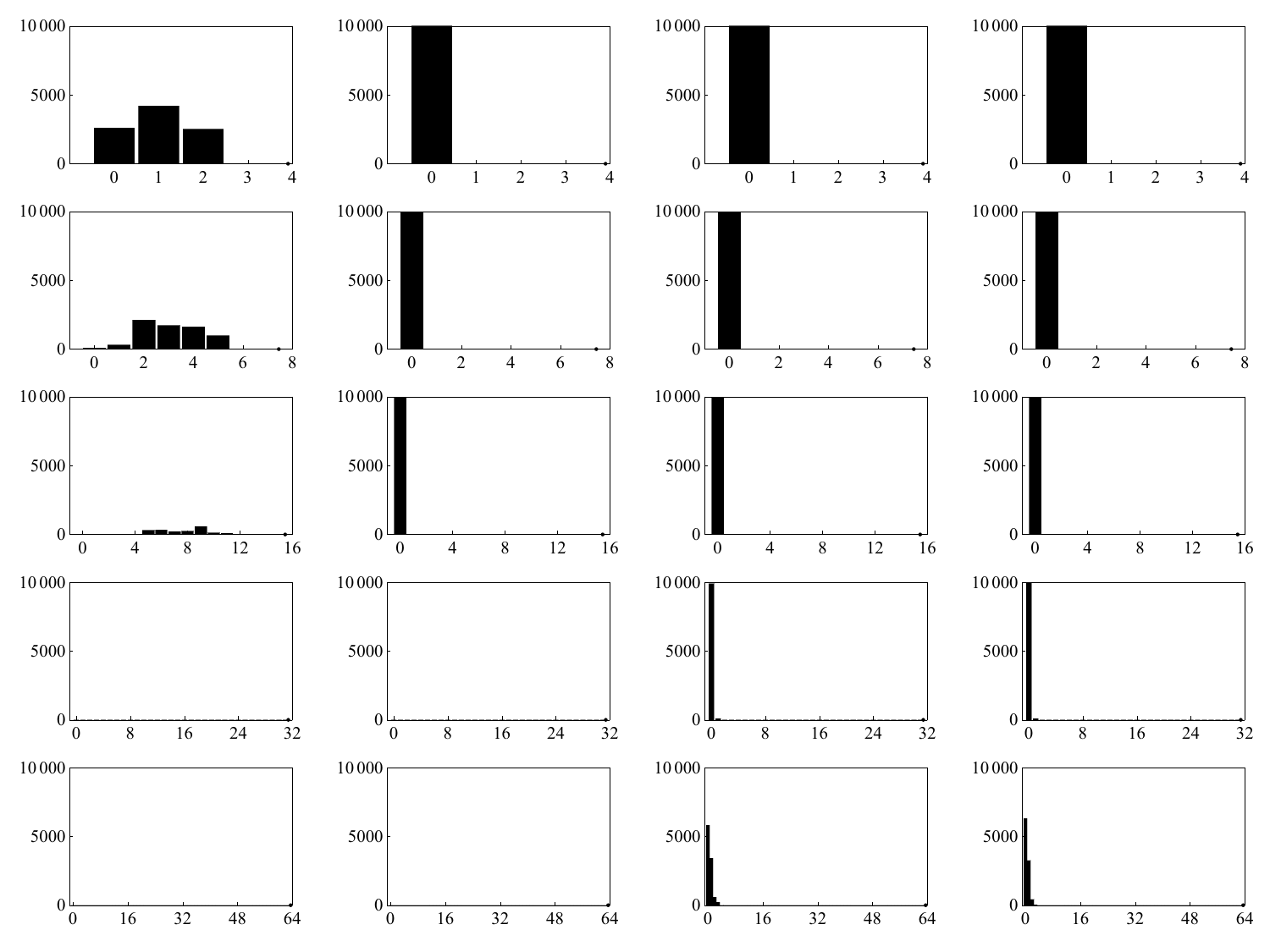}
\footnotesize
\put(8,76){ Uniform}
\put(30,76){$(\alpha,\sigma)=(-0.05,0.05)$}
\put(57,76){$(\alpha,\sigma)=(0,0.05)$}
\put(81,76){$(\alpha,\sigma)=(0.05,0.05)$}
\put(12,-1){Order $k$}
\put(36,-1){Order $k$}
\put(60,-1){Order $k$}
\put(85,-1){Order $k$}
\put(-2,68){\rotatebox{90}{$n=4$}}
\put(-2,50){\rotatebox{90}{$n=8$}}
\put(-2,35){\rotatebox{90}{$n=16$}}
\put(-2,22){\rotatebox{90}{$n=32$}}
\put(-2,6){\rotatebox{90}{$n=64$}}
\end{overpic}
\caption{\label{fig:statistics_SubCanonical_uniform} Statistics for partial-existence equilibria of order $k$ (rows) for $n=4,8,16,32,64$, using  interactions sampled from different distributions (columns): the uniform distribution, $\alpha_{ij}\sim\mathcal{U}_{[-1,1]}$, the truncated normal distribution $\mathcal{N}_{[-1,1]}$ with $\sigma=0.05$ and varying values $\alpha=-0.05,0,0.05$.
}
\end{figure}

\section{Mathematical analysis\label{app:math_analysis}}

In this section, we give the full details of the mathematical model. First, recall the (generalized) Lotka-Volterra equations (gLV), 
\begin{equation} \label{LV}
   x_i'(t)= r_i x_i(t) \left(1 - \frac{1}{K_i} \sum_{j =1}^n \al_{ij} x_j(t) \right),
\end{equation}
where $i = 1, \ldots , n$. The $K_i > 0$ are the carrying capacities and $r_i > 0$ the intrinsic growth rates. The numbers $\al_{ij}$, $i,j=1, \ldots, n$ form an $n \times n$-matrix $A$ with elements $(A)_{ij} = \al_{ij}$. Set $\al_{ii} = 1$, for all $i$. We also put $\oli{x} = (x_1, \ldots, x_n)$, and 
\[
F_i(x_1, \ldots, x_n) = 1 - \frac{1}{K_i} \sum_{j =1}^n \al_{ij} x_j.
\]

We are only interested in solutions in the {\em first orthant} 
\[
\R_+^n = \{ \oli{x} = (x_1, \ldots, x_n): x_j \geq 0, \text{ for all } 1 \leq j \leq n \}.
\]
We also speak of the {\em strict first orthant}, when all $x_j \geq 0$ are replaced with $x_j > 0$. 

If we add the effect of antibiotics as an additional linear term for each species $x_i=x_i(t)$ we get the following extended gLV-equations,
\begin{equation} \label{ext-gLV}
   x_i' = r_i x_i \left(1 - \frac{1}{K_i} \sum_{j =1}^n \al_{ij} x_j \right) - x_i l_i,
\end{equation}
where $l_i$ is the effect of the drug on taxon $i$. An {\em equilibrium} to the gLV-system (\ref{LV}) is a constant solution $\oli{x}(t) = \oli{x}^* = (x_1^*, \ldots, x_n^*) \in \R_+^n$. Obviously, an equilibrium satisfies $x_i^* = 0$ or $F_i(\oli{x}^*) = 0$ for each $i$. If we first put $l_i = 0$ for all $i$, consider the nullclines $F_i(\oli{x}) = 0$, or, equivalently,
\begin{equation} \label{hyperplane}
P_i: \al_{i1} x_1 + \al_{i2} x_2 + \ldots + \al_{in} x_n = K_i. 
\end{equation}
We can view $P_i$ as hyperplanes in $\R^n$. We will assume that the normal vectors of the hyperplanes are linearly independent, or, equivalently, that the matrix $A$ is invertible. This implies that there is a unique solution to the system of linear equations
\begin{equation} \label{coexistence}
A\oli{x} = \oli{K},
\end{equation}
where $\oli{K} = (K_1, K_2, \ldots, K_n)$. Of course this solution may not belong to the first orthant. An equilibrium is a {\em coexistence  equilibrium} if all $x_j > 0$, $ 1 \leq j \leq n$. So a coexistence equilibrium exists if and only if the solution to (\ref{coexistence}) belongs to the strict first orthant. Let $I \subset \{1, \ldots, n\}$ and put 
\[
\R_I^n = \{ (x_1, \ldots, x_n) \in \R^n : x_j \geq 0 \text{  for } j \in I \text{ and } x_j > 0 \text{ for } j \notin I \}.
\]
An equilibrium point $\oli{x}^* = (x_1^*, \ldots, x_n^*) \in \R_+^n$ is {\em stable} if there exists a neighbourhood $U$ around $\oli{x}^*$ such that for any starting point $\oli{y} \in U$, the solution $\oli{x}(t)$ to the gLV-system, with initial condition $\oli{x}(0) = \oli{y}$ has the property that $\oli{x}(t) \in U$, for all $t \geq 0$ and $\oli{x}(t) \raw \oli{x}^*$, as $t \raw \infty$.

From a probabilistic viewpoint, the computer experiments (see Appendix \ref{app:num_algorithm}) and quite recent results show that, under slightly different conditions on the matrix $A$, the probability of the existence of a coexistence equilibrium is quite low for large $n$ unless one imposes certain conditions on $A$, (see \cite{Clenet-Massol-Najim, Clenet-El-Najim, Bizeul-Najim}). We will consider invertible matrices $A$ that have a stable coexistence equilibrium. For the probabilistic results we also assume that $|\al_{ij}| < 1$, for $i \neq j$, and $\al_{ii} = 1$, for all $i$. The set of such $n \times n$-matrices is denoted by $\GG_n$. Let $\GG = \cup_n \GG_n$. If the system does have a stable equilibrium which is not a coexistence equilibrium, one can reason as follows. Suppose that those species $x_i$ that went extinct are those with indices $i \in I$. Then consider the following subspace in $\R^n$:
\[
\XX_I^n = \{ (x_1, \ldots,x_n) \in \R^n: x_i = 0, \text{ for } i \in I \}. 
\]
It is easy to see that an equilibrium $\oli{x}^* \in \R_+^n$ which also belongs to $\XX_I^n$  is stable in $\R^n$ if and only if it is stable in $\XX_I^n$ considered as a coexistence equilibrium (such equilibria are called partial-existence equilibria), and also that 
\[
F_i(\oli{x}^*) < 0, \text{ for all }  i \in I. 
\]
In a suitable neighbourhood $U$ of $\oli{x}^*$ where the above inequality holds, we have $x_i'(t) < 0$, for $i \in I$, and hence $x_{i}(t) \raw 0$, as $t\raw \infty$. Moreover, since $\oli{x}^*$ is stable also in $\XX_I^n$, we have $x_{i}(t) \raw x_i^*$ for $i \notin I$. 

Goh \cite{Goh} and Adachi, Takeuchi and Tokumaru \cite{Ad-Ta-To-1, Ad-Ta-To-2, Ad-Ta-1, Ad-Ta-2, Ad-Ta-3} gave some fundamental results on the stability of generalized Lotka-Volterra equations. 
\begin{Def}
    We say that the matrix $A$ belongs to the class $\SS$ if there is a diagonal matrix $D$ such that 
    \[
    DA + A^TD
    \]
    is positive definite. 
\end{Def}

This class of matrices is an example for which the gLV-system has one unique (global) equilibrium. In \cite{Goh} and \cite{Ad-Ta-To-1} the following is proved.
\begin{Thm}
    If $A$ belongs to $\SS$, then there is a unique globally stable equilibrium.
\end{Thm}

The converse of the above theorem is not true, and although there are many partial results, it seems that the question of sufficient and necessary conditions for a (global) equilibrium is not completely settled. 

If we let $\SS_n$ be the set of all ($n \times n$)-matrices in $\SS$, then, in particular, we have $\SS_n \subset \GG_n$. In the papers by Adachi, Takeuchi and Tokumaru, the following interesting class of matrices is also studied: 
\begin{Def}
We say that a quadratic matrix $A$ is called a $P$-matrix if all its principal minors are positive. The set of all $P$-matrices of order $n$ is denoted by $\PP_n$.
\end{Def}

If all off-diagonal elements of $A$ are non-positive, then Adachi, Takeuchi and Tokumaru prove:
\begin{Thm}
    If all $\al_{ij} \leq 0$ for $ i\neq j$, then there is a unique stable equilibrium if and only if $A$ is a $P$-matrix. 
\end{Thm}

The $P$-matrix condition is not necessary for stability, but the necessary condition for stability is, in a sense, quite close to the $P$-matrix condition.

A principal sub-matrix of an $n \times n$-matrix $A$ is a matrix where we have deleted a set $S \subset \{1, \ldots, n \}$ of rows and columns (the same set for both rows and columns) from $A$. In the following, we denote by $\DD_k(A)$ all principal sub-matrices of $A$ of size $k$. 
Let $D^*$ be the diagonal matrix with elements $(r_1/K_1)x_1^*, \ldots,(r_n/K_n)x_n^*$ along the diagonal (from the top left to the bottom right), and set $B = D^* A$. Then we have the following. 

\begin{Thm} \label{equi-s}
Suppose that there is a coexistence equilibrium for the gLV system. Let $B = D^* A$ as above. Then if the equilibrium is stable we must have
   \[
   \sum_{C \in \DD_k(B)} \det(C) > 0
   \] 
for all $k=0, \ldots n$.
\end{Thm}
As an immediate corollary, we get the following (putting $k=n$ in the above theorem).

\begin{Cor} \label{Cor-1}
    If a coexistence stable equilibrium exists, then $\det(A) > 0$. 
\end{Cor}

\begin{proof}[Proof of Theorem \ref{equi-s}]
Set
\[
F_i(x_1,\ldots,x_n) = 1 - \frac{1}{K_i} \sum_{j=1}^n \al_{ij} x_j.
\]
The Jacobian of the gLV-system becomes
\begin{equation} \label{J}
J = \left( \begin{array}{cccc} 
r_1(F_1-\frac{x_1}{K_1}) & -\frac{r_1}{K_1} \al_{12} x_1  & \ldots & -\frac{r_1}{K_1}\al_{1n}x_1 \\
\ldots & & & \\
-\frac{r_n}{K_n} \al_{n1}x_n & -\frac{r_n}{K_n} \al_{n2}x_n & \ldots & r_n(F_n - \frac{x_n}{K_n}) 
\end{array} \right). 
\end{equation}
   If we assume that $\oli{x}^*=A^{-1} \oli{K}$ is the coexistence equilibrium, then 
  $F_1(\oli{x}^*) = F_2(\oli{x}^*) = \ldots = F_n(\oli{x}^*) = 0$. Moreover, let us make a simple change of variables and let 
  \[
  y_i^* =  \frac{r_i}{K_i} x_i^* \quad \text{and} \quad  \oli{y}^* = (y_1^*, \ldots, y_n^*).
  \]
  Then, consequently, 
    \begin{align}
J &=  \left( \begin{array}{cccc} 
-y_1^* & -\al_{12}  y_1^*  & \ldots & -\al_{1n}  y_1^* \\
\ldots & & & \\
-\al_{n1}  y_n^* & -\al_{n2}  y_n^* & \ldots &  - y_n^* 
\end{array} \right) = -\oli{y}^* A. \\ 
\end{align}
We are looking now at the eigenvalues of this matrix. 

With $D_{i_1,i_2, \ldots,i_k}$ we mean the principal minor obtained by deleting all the rows and columns apart from $i_1, \ldots i_k$ from $A$. (We assume that all $i_1, \ldots, i_k$ are distinct.) We have 
    \begin{align}
p(\la) = \det(\la I - J) &= \la^n + \la^{n-1}(\sum_j y_j^*) + \la^{n-2} (\sum_{i \neq j} y_i^* y_j^* D_{ij})  \\
&+ \ldots  \la ( \sum_{i_1, \ldots, i_{n-1}} 
 y_{i_1}^* \ldots y_{i_{n-1}}^*   D_{i_1,\ldots, i_{n-1}}  ) + \det(A) y_1^* \ldots y_n^*.
    \end{align}
We have, by definition, that for each fixed $k$, 
    \[
    \sum_{i_1,\ldots,i_k} y_{i_1}^* \ldots y_{i_k}^* D_{i_1,\ldots,i_k} = \sum_{C \in \DD_k(B)} \det(C),
    \]
    which are precisely the coefficients in front of $\la^{n-k}$. By assumption, the real parts of the zeros of the polynomials are all negative, which is equivalent to saying that the system is stable. Since the polynomial is real, we can factor the polynomial in linear factors, for real zeros, and quadratic factors, for complex zeros appearing in conjugate pairs. We get  
    \begin{multline}
p(\la) = (\la - a_1)(\la - a_2)\cdot \ldots \cdot (\la - a_s) \\
\cdot (\la^2 - 2b_1z + b_1^2+c_1^2)(\la^2 - 2b_2z + b_2^2 + c_2^2)\cdot \ldots \cdot (\la^2 + 2b_tz - b_t^2 + c_t^2), \nonumber
    \end{multline}
    where each quadratic polynomial above does not have real roots. We see that if the systems is stable, all roots have negative real part, and hence all $a_j < 0$ for $j=1, \ldots, s$ and all $b_j < 0$ for $j=1, \ldots, t$. This implies that all the coefficients in the polynomial are positive, which is exactly the conclusion of the theorem.
\end{proof}

\begin{Rem}
There is a corresponding result for partial-existence equilibria, which we do not present here. It is, however, a similar straightforward calculation. 
    \end{Rem}

\begin{Rem}
    It is also easy to see that there is a neighbourhood $U$ of the identity interaction matrix $A=I$, so that the system has a stable coexistence equilibrium for $A \in U$. 
    \end{Rem}

We now give some elementary properties of a coexistence equilibrium, and later study how it changes under small perturbations. The results are based on standard linear algebra. See also \cite{Novak-etal-2026} for results on net effects of perturbations of a gLV system. 

If the equilibrium is a coexistence equilibrium, we can write it as
\[
\oli{x}^* = \bigcap_{i=1}^n P_i ,
\]
where $P_i$ are the hyperplanes (\ref{hyperplane}), 
given that the normal vectors to each hyperplane $P_i$ are linearly independent. 
We can also adjust the analysis for partial-existence equilibria by looking at the subspace $\XX_I^n$ defined above, but for simplicity, we assume that the equilibrium is a coexistence equilibrium.

Let $D_{ij}$ be the subdeterminant of matrix $A$ obtained by deleting row $i$ and column $j$. The cofactors are defined by $c_{ij} = (-1)^{i+j} D_{ij}$. Then the vector 
\[
\oli{u} = (c_{j1}, c_{j2}, \ldots, c_{jn}),
\]
is orthogonal to all other normal vectors of $P_i$, for $i \neq j$. 
Note also that $\oli{u} \cdot \oli{a}_j =\det(A)$. Clearly, $\oli{u}$ is the direction of the line 
\begin{equation} \label{line}
L_j = \bigcap_{i \neq j} P_i.
\end{equation}
Moreover, we can write the equilibrium as 
\[
\oli{x}^* = L_i \cap P_i. 
\]
This will be used in the following.

\begin{Def} Suppose that $\oli{x}^*=(x_1^*, \ldots, x_n^*) \in \R_+^n$ is an equilibrium. Then we say that the {\em order} of this equilibrium is $k$ if precisely $k$ coordinates $x_i^* = 0$. We call it a {\em partial-existence equilibrium} of order $k$. In other words, let $I$ be the set of indices $i$ such that $x_i^* = 0$. Then $|I|  = k$. 
\end{Def}

\begin{Prop}
    Let $A$ be the interaction matrix whose elements are denoted by $(A)_{ij} = \al_{ij}$. Suppose that $A$ is invertible and that $\alpha_{ii} = 1$ for all $i = 1, \ldots, n$. Then, if there is a stable coexistence equilibrium, there cannot be any stable partial-existence equilibrium of order $1$.
\end{Prop}
\begin{proof}
    Suppose $\oli{x}^*$ is a stable coexistence equilibrium. Then by Theorem \ref{Cor-1}, $\det(A) > 0$. We also know that $\oli{x}^*$ can be written as 
    \[
    \oli{x}^* = P_j \cap L_j,
    \]
    where $P_j$ are the hyperplanes $F_j=0$ and $L_j$ is defined in (\ref{line}). 
Let $\oli{n}_j = (\al_{j1},\ldots, \al_{jn})$ denote the normal vector to $P_j$. Since $\oli{u} \cdot \oli{n}_j = \det(A) > 0$ the angle between $\oli{n}_j$ and $\oli{u}$ is at most $\pi/2$. Suppose now that $c_{jj} > 0$, and that we have a partial-existence equilibrium $\oli{y}^* \in \{ x_j = 0\}$. Then $\oli{y}^*$ must be the intersection of the line $L_j$ with the coordinate plane $\{ x_j = 0\}$. Since $c_{jj} > 0$, the vector $\oli{u}$ whose foot is at $\oli{y}^*$ is pointing inward the positive orthant. Hence this partial-existence equilibrium must be on the same side of $P_j$ as the origin, i.e. so that $F_j(\oli{x}) > 0$. This makes it unstable in the $x_j$-direction (although it is stable in the subspace $\{ x_j = 0 \}$). If $c_{jj} < 0$ then we have the opposite situation. i.e. the partial-existence equilibrium $\oli{y}^*$ is unstable in $\{x_j = 0 \}$, but stable in the $x_j$-direction, because it lies on the other side of $P_j$ as compared to the origin, meaning $F_j(\oli{x}) < 0$. Hence, there cannot be any stable $1$st order partial-existence equilibrium.   
\end{proof}

We want to investigate how the coexistence equilibrium $\oli{x}^*$ changes under an external forcing (e.g. antibiotics). We assume that the forcing is made sufficiently small that the new equilibrium is also a coexistence equilibrium (however, this condition can be relaxed, see Section \ref{extinction}). If we force $x_j$ to change so that the vector $\oli{x}=(x_1,\ldots,x_n)$ still belongs to $L_j$, the (new) equilibrium slides along $L_j$, in the direction of $\oli{u}$. We see that if 
\[
\frac{c_{j1}}{c_{jj}} > 0,
\]
then decreasing taxon $j$ (with antibiotics) will reduce $x_1$ with speed $|c_{j1}/c_{jj}|$, and vice versa. Hence, we have the following (see also  \cite{Novak-etal-2026}, \cite{Nakajima-1992}).
\begin{Thm} \label{quotient}
Suppose there is a coexistence equilibrium. Changing $x_j$ with an amount $\de_j$, so that the other $x_k$ still belong to $P_k$ for $k \neq j$, implies that the $x_1$-coordinate changes with $\de_1$ according to the formula 
\[
\de_1 = \de_j \frac{c_{j1}}{c_{jj}}. 
\]
(We assume that $j \neq 1$.)
\end{Thm}

We now want to see how the coexistence equilibrium changes when the forcing stems from antibiotic pressure. We use the extended generalized Lotka-Volterra equation with an antibiotic, with spectrum vector $\oli{l} = (l_1, l_2, \ldots , l_n)$. It represents the drug effects on each taxon $j$. The resistant taxon is assumed to be $x_1$, so $l_1 = 0$ (the antibiotic has no effect on taxon 1). 
With the antibiotic, the nullclines $P_i$ then become
\[
P_i: \al_{i1} x_1 + \al_{i2} x_2 + \ldots + \al_{in} x_n = K_i - b_i, 
\]
where $b_i = \frac{l_i K_i}{r_i}$. 

Targeting a certain taxon $j$ with an antibiotic will increase $l_j$ and shift the corresponding nullcline in its normal direction. The equilibrium will consequently also shift. We are interested in how the $x_1$-coordinate of the equilibrium shifts if we increase $l_j$. 
More precisely, we are interested in the relationship between the antibiotic strength $l_j$ and $\de_j$, the change of $x_j$. Intuitively, $\de_j$ decreases when we target taxon $x_j$ with antibiotics. However, this is not necessary. 
As $l_j$ changes (increases), $x_j$ will decrease or increase depending on the direction of the line $L_j$.
A change in $b_j$, say $\De b_j$, will move the nullcline (hyperplane) $P_j$ a distance $d = \De b_j / \| \oli{t}_j \|$, where $t_j$ is the normal vector to $P_j$. If $\De b_j > 0$ the nullcline will move in the opposite direction of the normal vector $\oli{t}_j$. This change will induce a change $\De L$ along the line $L_j$, which is then projected onto the coordinate $x_j$. If $\th$ is the angle between $\oli{u}$ and $\oli{t}_j$, then if the change is 
$\De L (\oli{u}/ \| \oli{u} \|)$, 
\begin{equation} \label{L}
\De L = -\frac{\De b_j}{\| \oli{t}_j \| \cos \th} = - \De b_j \frac{\| \oli{u} \| \| \oli{t}_j \| }{\| \oli{t}_j \| (\oli{u} \cdot \oli{t}_j)} =  - \De b_j \frac{\| \oli{u} \|}{\det(A)}.
\end{equation}
Moreover, 
\begin{align}\de_j &= \De L \frac{(\oli{u} \cdot \oli{e}_j)}{\| \oli{u} \|} = \De L \frac{c_{jj}}{\| \oli{u} \|} \\
\de_1 &= \De L \frac{(\oli{u} \cdot \oli{e}_1)}{\| \oli{u} \|} = \De L
\frac{c_{j1}}{\| \oli{u} \|}.
\end{align}
Combining this with (\ref{L}), we get 
\begin{align}
\de_j &= -\De b_j \frac{c_{jj}}{\det(A)}, \label{inta-j}, \\ 
\de_1 &= -\De b_j \frac{c_{j1}}{\det(A)}. \label{inta-1}
\end{align}

In summary, if we put $\De b_j = K_j \la_j /r_j$, where $\la_j$ is the change in $l_j$, then 
\[
\de_j = - \frac{c_{jj}}{\det{A}} \frac{K_j \la_j}{r_j}, \qquad \de_1 = - \frac{c_{j1}}{\det{A}} \frac{K_j \la_j}{r_j} 
\]

 
We now study how $\de_1$ is affected by a change in several entries of the spectrum vector. Let us assume that we have two different antibiotics which have their ecological spectra 
 \[
 \La = (0,l_2,l_3, \ldots, l_n) \text{ and } \La' = (0,l_2',l_3', \ldots, l_n')
 \]
 (we put $l_1 = l_1' = 0$, since it is assumed that the taxon $x_1$ is resistant). 

We say that $\La'$ has a {\em (strictly) stronger spectrum} than $\La$ if $l_j' \geq l_j$ for all $j$, with at least one strict inequality and write $\La' > \La$. Suppose we increase $l_j$ with $\la_j$ for a subset of $j$:s, say $I \subset [2, \ldots]$. Then, using (\ref{inta-1}) and linearity,  the change in taxon $1$ becomes
\begin{equation} \label{delta-1-sum}
\de_1 = \sum_{j \in I} -\frac{K_j \la_j}{r_j} \frac{c_{j1}}{\det(A)}. 
\end{equation}

\subsection{Extinction of taxa} \label{extinction}

We can adjust the model to the situation when one or more taxa go extinct under antibiotic pressure. Let us start with a coexistence equilibrium as above with $\oli{x}^* = (x_1^*, \ldots, x_n^*) \in \R^n$, where all $x_j^* > 0.$ Let us assume that the community is exposed to an antibiotic with a spectrum vector $\La = (0,l_2,l_3, \ldots, l_n).$, starting from no treatment, i.e. from the spectrum vector $\oli{0}=(0,0,\ldots, 0)$. Let $I$ be the set of non-zero entries in $\La = (0,l_2,l_3, \ldots, l_n).$ 
We will study the equilibrium for a spectrum vector $t \La$ as $t$ increases from $0$ to $1$. We define a sequence of $t_1 < t_2 <\ldots < t_l$ as follows. Let $t_1 \in (0,1]$ be the smallest number such that there is a coexistence equilibrium for the spectrum vector $t \La$ for $t < t_1$ and for $t=t_1$ at least one taxon has gone extinct. Suppose that $I_1$ are those taxa that survive at $t=t_1$, hence those that went extinct are those with index $j \in I_1^c$. Then the change in abundance for the resistant taxon $x_1$ is
\[
\de_1 = \sum_{j=2}^n -\frac{K_j \la_j}{r_j} \frac{c_{j1}}{\det(A)},
\]
which is the same as formula (\ref{delta-1-sum}) since $l_j = 0$ if $j \notin I$. The next number $t_2 > t_1$ is defined as the smallest number larger than $t_1$, and less than or equal to $1$, for which at least one more taxon goes extinct in addition to those of $I_1^c$. In other words, for $t_1 \leq t < t_2$ there are precisely those taxa in $I_1$ that are present and the other taxa $x_j$, $j \in I_1^c$, are extinct. We let $I_2$ be those taxa that are present at $t=t_2$. Hence $I_1 \supset I_2$, with strict inclusion. We then define $t_3 > t_2$ to be the smallest number, less than or equal to $1$, so that at least one more taxon in addition to those in $I_2^c$ goes extinct. So for $t_2 \leq t < t_3$ precisely those taxa in $I_2$ are still present and no other taxa. We continue in this way as long as possible, and obtain a sequence of numbers $t_1 < t_2 \ldots < t_l \leq 1$ and corresponding index sets 
\[ 
I_1 \supset I_2 \supset I_3 \supset \ldots \supset I_l, 
\]
with strict inclusions, so that precisely those taxa in $I_l$ are still present for $t_1 \leq t \leq 1$. Between each adjacent pair of $t_j$'s, i.e. in each interval $[t_k,t_{k+1})$, precisely those taxa in $I_k$ are present. Moreover, during the change of $t$ each such interval, when $t=t_{k+1}$, we have a change in abundance for the resistant taxon $x_1$, that corresponds to the formula (\ref{delta-1-sum}), but with $A$ replaced by a smaller matrix where all rows and columns $j \in I_k^c$ are removed. Let $A_k$ be this matrix and $c^{(k)}_{ij}$ be its corresponding cofactors. Then 
\[
\de_{1k} = \sum_{j \in I_k} -\frac{K_j \la_j}{r_j} \frac{c_{j1}^{(k)}}{\det(A_k)}
\]
is the change in taxon $x_1$ when we change the spectrum vector from $t_j\La$ to $t_{j+1} \La$. The total change in taxon $x_1$ is then
\[
\de_1 = \sum_{k=1}^l \de_{1k} = \sum_{k=1}^l \sum_{j \in I_k} -\frac{K_j \la_j}{r_j} \frac{c_{j1}^{(k)}}{\det(A_k)}.
\]

\comm{
\subsection{Rescaling}
We have as a standing assumption that a coexistence stable equilibrium exists, we must have that $\det(A) > 0$ according to Corollary \ref{Cor-1}. We now make use of the fact that the LV-system is scaling invariant in the following sense. Let $c > 0$ and put 
\[
y_i(t) = c x_i(t), \text{ for } 1 \leq i \leq n.   
\]
Then the extended LV-system becomes
\[
y_i'(t) = r_i y_i(t) (1 - \frac{1}{K_i} \sum_{j=1}^n \frac{\al_{ij}}{c}y_j(t) )  - y_i(t) l_i,
\]
where $\be_{ij} = \al_{ij}/c$. 
Hence in the new system $r_i$, $K_i$, $l_i$ all are the same, but the matrix $A$ has been rescaled. Let $B = (1/c) A$. We can adjust $c > 0$ so that $\det(B) = 1$. 
}

\subsection{Probability distributions}
We are now interested in the expected outcome of the $\de_1$ in (\ref{delta-1-sum}). Since the expectation is linear, 
\[
\E(\de_1) = \E \left( \sum_{j \in I} -\frac{K_j \la_j}{r_j} \frac{c_{j1}}{\det(A)} \right) = 
- \sum_{j \in I} \frac{K_j \la_j}{r_j} \E \left( \frac{c_{j1}}{\det(A)} \right) .
\]
Our goal is to estimate the expectations $\E(c_{j1}/\det(A))$, given some assumptions on the $\al_{ij}$ as random variables.

\begin{Lem} \label{main-1}
    Assume that $-1 < \al_{ij} <1$ for $i \neq j$ ($\al_{ii} = 1$) are independently identically distributed, and $\E(\al_{ij}) = \al$. Then, for each $j \neq 1$,  
    \begin{itemize}
        \item $\E(c_{j1}) = -\al(1-\al)^{n-2}$ 
        \item $\E(c_{jj}) = (1-\al)^{n-2}(1 + \al (n-2))$.
    \end{itemize}
\end{Lem}

\begin{proof}
The determinant of the minors $D_{ij}$ is a polynomial in the variables $\al_{ij}$. Let us denote it by $p_n(\al_{11}, \al_{12}, \ldots, \al_{nn})$, or simply by $p_n([\al_{ij}])$. Since all $\al_{ij}$ are assumed to be independent, and since the expectation is equal to $\al$, we have 
    \[
    \E(\text{det}(B)) = \E (p_n([\al_{ij}])) = p_n([\E(\al_{ij}])=p_n(\al, \ldots, \al).
    \]
    By simple Gaussian elimination, 
    \begin{align}
\E(\text{det}(B)) &= \left| \begin{array}{ccc}
1 & \al & \al \ldots \al \\
\al & 1 & \al \ldots \al \\
& \ldots & \\
\al & \al & \al \ldots  1 
\end{array} \right| 
= 
\left| \begin{array}{ccc}
1 & \al \ldots & \al \\
\al-1 & 1-\al & 0  \ldots 0 \\
& \ldots & \\
\al-1 & 0 & 0 \ldots  1-\al 
\end{array} \right| 
\\
&= (1-\al)^{n-2} - \al(\al-1)(1-\al)^{n-3}(n-2) = (1-\al)^{n-2} (1  + \al(n-2)). 
    \end{align}
    The second statement is proved in a similar way. 
    We continue with computing $D_{n1}$,
    \begin{align}
D_{n1} &=     \left| \begin{array}{ccc}
\al & \al \ldots & \al \\
1 & \al \ldots & \al  \\
& \ldots & \\
\al & \ldots 1 & \al 
\end{array} \right| 
=
\left| \begin{array}{cccc}
\al & \al \ldots & & \al \\
1-\al & 0 \ldots & & 0  \\
0 & 1-\al & \ldots & 0 \\
0 & 0 \ldots & 1-\al & 0 
\end{array} \right| \\
&= \al (1-\al)^{n-2}(-1)^{n-2}.
    \end{align}
By the definition of cofactors, 
      \[
    c_{n1} = (-1)^{n+1}D_{n1}  = - \al (1- \al)^{n-2}, 
    \]
    which finally proves the lemma. 
\end{proof}
\begin{Rem}
We can of course adjust the lemma and conclude that, for an $n \times n$ interaction matrix $A$, we have  
\[
\E(\det(A)) = (1-\al)^{n-1} (1 + \al (n-1)).
\]
\end{Rem}
We see that at the singularity $\al_s = -1/(n-1)$ (apart from $\al = 1$), the expectation vanishes. It seems that for $\al < \al_s$ the LV-system blows up, i.e., the population (of at least one taxon) grows to infinity (see Table \ref{tab:statistics_attractors_normal} for $\al = -0.05$ and $n=32,64$). If $\det(A)$ is close to zero it can cause problems when computing the expectation of the quotient $c_{j1}/\det(A)$. For a coexistence equilibrium to exist we have necessarily, by \cite{AspenbergMartensWaldetoft}, that $\det(A) > 0$. We want stay away from the singular matrices and assume that $\det(A) \geq \ka$ for some fixed $\ka > 0$.

\begin{Lem} \label{expect-order}
   We have the following estimates, where $\si(\al)$ is the standard deviation of $\al_{ij}$: 
    \begin{itemize}
        \item $\E(c_{j1}^2) = \si(\al)^2 + \al^2 + \OO(\al^2 \si(\al)^2) + \OO(\al^3)$
        \item $\E(\det(A)^2) =  1 - n(n-1) \al^2  +  \OO(\al^2  (\al^2 + \si(\al))^2) + \OO(\al^3)$
        \item $\si(\det(A))^2  = 2n(n-1) \al^2 + \OO( \al^2 (\al^2+ \si(\al))^2) + \OO(\al^3)$.
    \end{itemize} 
\end{Lem}
\begin{proof}
Without loss of generality, suppose that $j=n$. If we look at the submatrix by deleting row $n \neq 1$ and column $1$, we get 
\begin{equation} \label{ko}
c_{n1} = \al_{1n} + S,
\end{equation}
where $S$ is a sum of all products with at least $2$ elements of the form $\al_{ij}$. We now make the observation that the terms with exactly $2$ factors cannot contain $\al_{1n}$ (since if $\al_{1n}$ appears in the product then one has to take at least two more $\al_{ij}$'s from the matrix where row $1$ and column $n$ is deleted; such a matrix has $1$s on the diagonal). 

If we square the expression we get 
\[
c_{n1}^2 = \al_{1n}^2 + S^2 + 2 \al_{1n} S. 
\]
Each term in $S^2$ is a sum of elements of the type 
\begin{equation} \label{al-factors}
\al_{i_1,j_1} \ldots \al_{i_p,j_p} \al_{k_1,l_1}^2 \ldots \al_{k_q,l_q}^2, 
\end{equation}
where the collection of all pairs $(i_s,j_s)$, $(k_t, l_t)$, where $1 \leq s \leq p$ and $1 \leq t \leq q$, does not contain two pairs that are equal. There are some restrictions. We have $p \geq 2$ or $p=0$ (because of the combinatorics), and moreover; if $p=0$ then $q \geq 2$, and if $q=0$ (no quadratic terms) then $p \geq 4$. The last term, $2 \al_{1n} S$ consists of elements of the same type but with the restriction that the only square that can appear is $\al_{1n}^2$, and if $\al_{1n}^2$ appears then $p \geq 2$ (according to the observation on $S$ above). Hence either $p \geq 2$ and $q=1$, or $p \geq 3$ and $q=0$. 

Moreover, all the factors in (\ref{al-factors}) are independent as random variables. For each factor of the type $\al_{k_t,l_t}^2$ we use the identity 
\[
\E(\al_{ij}^2) = \si(\al)^2 + \E(\al_{ij})^2 = \si(\al)^2 + \al^2. 
\]
Hence we get, taking into account the different possibilities of $p$ and $q$, 
\begin{multline}
\E(c_{n1}^2) = \si(\al)^2 + \al^2 + \OO(\al^2(\al^2 + \si(\al)^2)) + \OO(\al^3) 
\\
= \al^2 + \si(\al)^2 + \OO(\al^2 \si(\al)^2) + \OO(\al^3). 
    \end{multline}
    
If we turn to the determinant of $A$, we get, with a similar argument, 
\[
\det(A) = 1 - \sum_{i < j} \al_{ij} \al_{ji} + T,
\]
where $T$ is a sum of products of at least $3$ distinct factors $\al_{ij}$. We have,

\begin{multline}
(\det(A))^2 = (1 - \sum_{i < j} \al_{ij}\al_{ji} + T)^2 \\
= 1 - 2 \sum_{i < j} \al_{ij}\al_{ji} + \big( \sum_{i < j} \al_{ij} \al_{ji} \big)^2 +  2T + T^2  - 2 T \sum_{i < j} \al_{ij} \al_{ji} \\
= 1 - 2 \sum_{i < j} \al_{ij}\al_{ji} + \sum_{i < j} \al_{ij}^2 \al_{ji}^2 + \sum_{i < j, k < l} \al_{ij}\al_{ji} \al_{kl} \al_{lk} + 2T + T^2 - 2T \sum_{i < j} \al_{ij}\al_{ji}, 
\end{multline}
where all factors in the third sum on the last row are distinct. Similar to the first case, the terms $T^2, 2T$ and $2T \sum_{i < j} \al_{ij} \al_{ji}$ can only contain products of the form (\ref{al-factors}). For the term $T^2$ we have that if $p=0$ then $q \geq 3$ and if $p=2$ then $q \geq 2$. Moreover, as in the first case, one can conclude that $p \neq 1$. Indeed, if $p=1$ then we are considering the product of two different terms from $T$, which have all factors identical but one single additional factor in one of the terms. But since $A$ has $1$s on the diagonal, it is easy to see that this is impossible.

For the last term, $-2T \sum_{i < j} \al_{ij} \al_{ji}$, we must have $p > 0$ (because the number of quadratic factors can be at most $2$, and $T$ has at least $3$ distinct factors). 
Moreover, if $p=1$ there must be exactly $2$ quadratic factors, i.e. $q=2$, and these factors must be of the form $\al_{ij}^2$ and $\al_{ji}^2$ (i.e. symmetric). However the product from $T$ contains another single factor. Since the diagonal of $A$ consists of just $1$s, we arrive at a contradiction similar to the case just above for $T^2$. Hence $p \neq 1$, and in fact $p \geq 2$. Moreover, if $p = 2$ then $q = 2$ (because $q=0$ gives $p\geq 5$ and $q=1$ gives $p \geq 3$). 

Since $\E(\al_{ij}^2 \al_{ji}^2) = \E(\al_{ij}^2) \E(\al_{ji}^2)$, and using again that $\E(\al_{ij}^2) = \si(\al)^2 + \E(\al_{ij})^2$, 
\begin{multline}
\E(\det(A)^2) = 1 - n(n-1) \al^2  + \OO((\al^2 + \si(\al)^2)^2) + \OO(\al^3) \\
+ \OO(\al^2 (\al^2 + \si(\al)^2)^2) + \OO((\al^2 + \si(\al)^2)^3) \\  
=  1 - n(n-1) \al^2  +  \OO(\al^2 (\al^2 + \si(\al)^2)^2) + \OO(\al^3). 
\end{multline}
Therefore, using Lemma \ref{main-1},  
\begin{multline}
\E(\det(A)^2) - \E(\det(A))^2 \\
= 1 - n(n-1) \al^2 + \OO(\al^2(\al^2 + \si(\al)^2)^2) + \OO(\al^3) -  ((1-\al)^{n-1}(1 + \al (n-1)))^2 \\
= 1 - n(n-1) \al^2 + \OO(\al^2(\al^2 + \si(\al)^2)^2) + \OO(\al^3) \\
- (1 - (n-1)\al + \frac{(n-1)(n-2)}{2} \al^2 + \OO(\al^3))^2 (1 + 2 \al (n-1) + \al^2(n-1)^2) \\
= 1 - n(n-1) \al^2 + \OO(\al^2 (\al^2 + \si(\al)^2)^2) + \OO(\al^3) \\
- (1 - 2\al(n-1) + (n-1)^2 \al^2 + (n-1)(n-2)\al^2 + \OO(\al^3))(1 + 2\al(n-1) + \al^2(n-1)^2) \\
= 1 - n(n-1) \al^2 + \OO(\al^2 (\al^2 + \si(\al)^2)^2) + \OO(\al^3) \\
-(1 + \al^2 n(n-1) + \OO(\al^3) ) \\
=2n(n-1) \al^2 + \OO(\al^2 (\al^2 + \si(\al)^2)^2) + \OO(\al^3).
\end{multline}

\end{proof}

 Let $X=c_{j1}$ and $Y = \det(A)$. Then $X$ and $Y$ are random variables, depending on the underlying probability space. We will possibly change these random variables a little and consider the restricted variables $X=(c_{j1} | \det(A) \geq \ka )$ and $Y = (\det(A) | \det(A) \geq \ka)$ in order to have the determinant of $A$ bounded away from zero (this is necessary to have a stable coexistence equiibrium). 
\begin{Thm} \label{main-2}
Let $\ka > 0$, and suppose $A$ is an $n \times n$-matrix. Suppose that $\al_{ij}$ are i.i.d:s and restricted to the set of matrices where $\det(A) \geq \ka > 0$. Let $\E(\al_{ij}) = \al$. We assume that $\al$ and $\si(\al)$ are so that the conditional distribution of $\det(A) \geq \ka$ is sufficiently close to that of the distribution without this restriction (so $\al$ depends on $\ka$ and $\si(\al) \leq \OO(\al)$ is sufficiently small).  Put $X=(c_{j1} | \det(A) \geq \ka )$ and $Y = (\det(A) | \det(A) \geq \ka)$. Then, if $a = \E(X)$ and $b = \E(Y)$, 
\[
\E \left( \frac{X}{Y} \right) = \frac{a}{b} + R_2(X,Y), 
\]
where $|R_2(X,Y)| = \OO(\al^2)$. 
\end{Thm}

\begin{proof}
    We use the Taylor expansion around the point of expectations $(a,b)$;
    \[
    f(X,Y) = \frac{X}{Y} = \frac{a}{b} + \frac{1}{b}(X-a) - \frac{a}{b^2}(Y-b) - \frac{1}{b^2}(X-a)(Y-b) + \frac{2}{b^3}(Y-b)^2 +  h.o.t.
    \]
   If we use the Taylor series of $f$ order $2$ with Lagrange's remainder around the point of expectations $(a,b)$, we get 
    \begin{equation}
    \E \big(\frac{X}{Y} \big) = \frac{a}{b} + \frac{1}{b}(\E(X)-a)  - \frac{a}{b^2}(\E(Y) - b)+ \E(R_2(X,Y)) = \frac{a}{b} + \E(R_2(X,Y)),
\end{equation} 
since $f_{xx}'' =0$. Let us put $h = X-a$ and $k=Y - b$. The remainder becomes
\begin{multline}
R_2(X,Y) = \int_0^1 2(1-s) f_{xy}''(a + sh,b + sk) hk + f_{yy}''(a + sh,b + sk) k^2 \ud s \\
= \int_0^1 (1-s) \left( \frac{-2}{(b + sk)^2} hk + \frac{2(a + sh)}{(b + sk)^3)} k^2  \right) \ud s \\
= \int_0^1 (1-s) \frac{2k (ay - bx)}{(b+ sk)^3}\ud s .
\end{multline}
Using the change of variables $u = b + sk$ gives,
\begin{equation} \label{remainder}
R_2(X,Y) = \frac{2 (aY - bX)}{k} \int_b^Y \frac{Y}{u^3} - \frac{1}{u^2} \ud u = \frac{(aY-bX)(Y-b)}{Yb^2}.
\end{equation}
Now we want to compute the expectation of the remainder. 
We have, by (\ref{remainder}), 
\[
R_2 = \frac{(aY - bX)(Y-b)}{Yb^2} = \frac{(a(Y-b) - b(X-a))(Y-b)}{Yb^2}
\]
and therefore
\begin{equation} \label{expect-r2}
|\E(R_2)|\leq \frac{|a| \si(Y)^2 + |b| \si(X) \si(Y)}{\ka b^2}. 
\end {equation}
By Lemma \ref{expect-order}, $\si(X)^2 = \E(X^2) - \E(X)^2 = \OO(\al^2)$, if $\si(\al) = \OO(\al)$, and $\si(Y)^2 = \OO(\al^2)$. A simple computation shows that $b = \OO(1)$.  Hence, if $\al$ is small enough,
\[
| \E(R_2) | \leq \frac{|a| \OO(\al^2) + b \OO(\al) \OO(\al) }{\ka b^2} = \OO(\al^2).
\]

\comm{
By the Cauchy-Schwarz inequality
\[
|\E(R_2)| \leq \sqrt{\E\frac{(aY-bX)^2}{(Yb^2)^2}} \sqrt{(\E(Y-b)^2)}.
\]
Recall that we assume that $Y \geq \ka > 0$. The expression on the very right under the square root is the variance $\si(Y)^2$. Since the numerator is non-negative and the denominator bounded from below by $b^2 \ka$, the left expression under the square root is bounded by
\[
\E(a^2 Y^2 - 2abXY + b^2X^2) \E((Yb^2)^{-1}) \leq \frac{a^2 \E(Y^2) + b^2 \E(X^2) - 2ab \E(XY)}{\ka b^2}. 
\]
Finally, 
\[
\E(|XY|) \leq \sqrt{\E(X^2)} \sqrt{\E(Y^2)}.
\]
An easy computation shows that $\E(Y^2) = \OO(1)$. By Lemma \ref{expect-order} $\E(X^2) = \OO(\al^2)$ (since $\si(\al)$ is bounded by $\OO(\al)$; we also take into account that the restricted distribution is very close to the unrestricted distribution). So,
\[
|\E(XY)| \leq \OO(\al). 
\]
Moreover, by Lemma \ref{expect-order}
\[
\si(Y)^2 = \E(Y^2) - \E(Y)^2 = \OO(\al^2) \qquad \si(X)^2 = \E(X^2) - \E(X)^2 = \OO(\al^2).
\]
By Lemma \ref{main-1}, $a = -\al(1-\al)^{n-2}$ and $b \geq (1-\al)^{n-2}(1 + \al(n-2))$. Since $a = \OO(\al)$, we get
\begin{multline}
|\E(R_2)| \leq  \sqrt{\E\frac{(aY-bX)^2}{(Yb^2)^2}} \sqrt{ \E((Y-b)^2)} \\
= \frac{1}{\ka b^2} \sqrt{a^2 \OO(1) + b^2 \OO(\al^2) -2 ab \OO(\al)} \sqrt{\OO(\al^2)} = \OO(\al^2). \label{last}
\end{multline}
}
\end{proof}

A few remarks are in order. Since $a = \OO(\al)$, the above result means that $a/b$ is a good approximation of the expectation of $X/Y$, if $\al$ and $\si(\al)$ are small:   
\[
\E \left( \frac{X}{Y} \right) = \frac{a}{b}(1 + \OO(\al)). 
\]
We do not know the probability distribution of $X$ and $Y$ given that $\det(A) \geq \ka > 0$. However, we conjecture that the remainder $R_2$ is small for a rather large range of $\al$ (supported by the numerical experiments), given that $\si(\al)$ is of the same order as $\al$, i.e. so that
\[
\E \left( \frac{X}{Y} \right) \approx \frac{a}{b}, 
\]
is good, or at least not so bad, for many more values of $\al$ than indicated in the theorem. Obviously, by (\ref{expect-r2}), if $\si(X)$ and $\si(Y)$ are sufficiently small, then the results in the above theorem holds for any $\al \in (-1,1)$ (not just for small $\al$). 

\comm{
Let us estimate the remainder in another way. We have, by (\ref{remainder}), 
\[
R_2 = \frac{(aY - bX)(Y-b)}{Yb^2} = \frac{(a(Y-b) - b(X-a))(Y-b)}{Yb^2}
\]
and therefore
\[
|\E(R_2)|\leq \frac{|a| \si(Y)^2 + |b| \si(X) \si(Y)}{\ka b^2}. 
\]
So if we can bound the standard deviations of the distributions $X$ and $Y$ we can get better results for all $\al \in (-1,1)$ (not just for small $\al$). 
}

One can estimate the order of the remainder $R_2$ more precisely, but we leave that for future works. Instead we refer to the numerical computations in Appendix \ref{app:num_algorithm}. However, heuristically, we want to think of $\E(c_{j1})/\E(\det(A))$ as a first order approximation of $\E(c_{j1}/\det(A))$. If this is the case we arrive at the following.  
\begin{Prop}
Given the assumptions in the above theorem \ref{main-2}, for $\al$  and $\si(\al)$ sufficiently small, or, alternatively, if the standard deviations of $X$ and $Y$ are sufficiently small, we have, if all $\la _j \geq 0$ and at least one $\la_j > 0$,
      \begin{itemize}
        \item if $0 < \al < 1$, then $\E(\frac{c_{j1}}{\det(A)}) < 0$, and $\E(\de_1) > 0$
        \item if $-1 < \al < 0$, then $\E(\frac{c_{j1}}{\det(A)}) > 0$, and $\E(\de_1) < 0$,
      \item if $\al = 0$, then $\E(\frac{c_{j1}}{\det(A)}) = 0$, and $\E(\de_1) = 0$.
    \end{itemize}
\end{Prop}

\subsection{Combinatorics in the community network}  \label{combinatorics}

We now give a combinatorial interpretation of the interactions between the $n$ taxa. It is well-known that one can consider the matrix $A$ as a graph with weights between $n$ vertices. Let $A$ be a given interaction $n \times n$-matrix with entries $\al_{ij}$. For simplicity, and as in the probabilistic part, we assume that $\al_{ii} = 1$, for all $i$ and $|\al_{ij}| < 1$ for $i \neq j$. The number $\al_{ij}$ can be considered as a weight of an {\em arc} from vertex $i$ to vertex $j$. A subgraph $\Ga$ of order $k$ of $A$ is a collection of $k$ vertices ($k \leq n$) that has the property that each vertex in $\Ga$ has precisely one outgoing arc and precisely one incoming arc. A subgraph $\Ga$ can then be interpreted as a collection of closed paths as follows. Let $v_0$ be a vertex in $\Ga$. It has precisely one outgoing arc. If we follow that arc then we end up in another vertex $v_1$. This vertex also has one outgoing arc, so following this one we end up in another vertex $v_2$. Continuing in the same way, finally we have a vertex $v_{k-1}$ which lands at $v_0$ (it cannot land at any other vertex already visited, because that would create a vertex that has two incoming arcs). This creates a closed path (also called a loop) $p=(v_0, \ldots, v_{k-1}, v_k)$, starting at $v_0$ and ending at $v_k= v_0$. We associate a {\em weight} to this path as
\[
\om(p) = (-1)^g \al_{v_k,v_{k-1}} \al_{v_{k-1}, v_{k-2}} \ldots \al_{v_1,v_0},
\]
where $g$ is the sign of the permutation, taking $p=(v_0, \ldots, v_k)$ to its ordered path $\si(p) = (w_1, \ldots, w_k)$, where $w_1 < w_2 < \ldots < w_k$. Note that the trivial loops $\al_{ii}$ all have weights equal to $1$ (they correspond to the diagonal elements of $A$). All other loops avoid the diagonal elements of $A$. 
The weight of a subgraph $\Ga$ is the sum of all weights of all loops:
\[
\om(\Ga) = \sum_{p \in \Ga}  \om(p).
\]
The determinant of an $n \times n$-matrix $A$ is then the sum of all weights of all subgraphs of order $n$. However, since the diagonal elements are equal to $1$ we can then write the determinant of $A$ as the sum of the weights of all subgraphs of $A$:
\[
\det(A) = \sum_{\Ga} \om(\Ga).
\]

Consider a path $p=(j_0, \ldots, j_k)$ (recall that all the $j_l$'s are distinct). Set 
\[
\al_p = \al_{j_kj_{k-1}} \ldots \al_{j_1j_0}.
\]
Let $E_p$ be the determinant of the matrix obtained by deleting the rows and columns indexed by $i$ if $i$ appears in the path $p$. We define $E_p$ as $1$ if the matrix obtained in this way becomes the empty matrix. Let $\PP_n$ be the set of all possible paths $(j_0, \ldots, j_k)$ starting from some $j_0$ and ending at some $j_k$ ($k \leq n$). Then we have the following lemma:
\begin{Lem} \label{det-rep}
Fix $i,j \leq n$. Let $\PP_{n,i,j} \subset \PP_n$ be the paths in $\PP_n$ that start at $i$ and end at $j$. Then the determinant $D_{ij}$ of the matrix obtained by deleting row $i$ and column $j$ from $A$ can be written
\[
D_{ij} = \sum_{p \in \PP_{n,i,j}} \al_p E_p,
\]
where $E_p$ is a sum of products of loops disjoint from $p$. 
\end{Lem}

{\bf Example.} If $p = (1, 3, 4)$,  and $n=6$, so $A$ is a $6 \times 6$ matrix, then
\[
E_p = \left| \begin{array}{ccc}
 1 & \al_{25} & \al_{26} \\
 \al_{52} & 1 & \al_{56} \\
  \al_{62} & \al_{65} & 1
\end{array} \right| = 1 - \al_{25} \al_{52} - \al_{26} \al_{62} - \al_{65} \al_{56} + \al_{62} \al_{25} \al_{56} + \al_{63} \al_{35} \al_{56}.
\]

In the example above all the terms in $E_p$ are of the form $\al_{p'}$ where $p'$ is a closed path. 

\begin{proof}
Without loss of generality, put $i=n$ and $j=1$. Let $p= (j_0,j_1, \ldots, j_k)$ be a path starting at $j_0=n$ ending at $j_k=1$. Then 
\[
\al_p = \al_{1,j_{k-1}} \ldots \al_{j_2,j_1} \al_{j_1,n}, 
\]
where $j_0 = n$, and $j_1 \neq n$. Note that all entries $\al_{j_{i+1},j_i}$ in the above path are all contained in the matrix obtained by deleting row $n$ and column $1$. 

Let $M_p$ be the matrix obtained by deleting the rows and columns $j$ from $A$, whenever $j \in p$. It is easy to see that the diagonal elements in $M_p$ are $(M_p)_{ii} = 1$. Moreover, we have that if $(M_p)_{ij} = \al_{i'j'}$, then $(M_p)_{ji} = \al_{j'i'}$. 

Clearly, the determinant $D_{n1}$ is a sum of products of the form $\al_p M_p$ as in the lemma. The only thing left to prove is that every term in $M_p$ is a product of closed paths (disjoint from $p$). 

Instead of working with $\al_{ij}$, let us rename the entries of the matrix $M_p$ so that $(M_p)_{ij} = \be_{ij}$. Suppose that $M_p$ has $m$ rows and columns (of course it is a square matrix). Then it is easy to see that each term, apart from the diagonal term $\be_{11}\be_{22} \ldots \be_{mm} = 1$, in the determinant of $M_p$ is a term of the form 
\begin{equation} \label{beta}
\be_{l_1 \si(l_1)} \ldots \be_{l_k \si(l_k)},
\end{equation}
for some $k \leq m$, where $\si$ is a permutation of a subset $\{l_1, \ldots, l_k\} \subset \{1, \ldots, m\}$, and where $\be_{j\si(j)} \neq 1$. 
Hence we consider derangements. For each such derangement, it is clear that we can order the factors $\be_{j \si(j)}$ in order so that, starting with $\be_{l_1 \si(l_1)}$, we adjoin it with $\be_{\si(l_1), \si(\si(l_1))}$, which in turn is adjoined with $\be_{\si^2(l_1), \si^3(l_1)}$ and so on. The orbit $\si^j(l_1)$ has to terminate, which creates a loop. It is easy to see that each such loop corresponds to a cycle in the permutation $\si$ (it is a well-known result in combinatorics that every permutation can be written as a product of loops). In other words, there has to be an index $b$ such that $\si^b(l_1) = l_1$ (of course we can have $b =k$ and then we have a full loop). So each loop will look like 
\[
\be_{l_1, \si(l_1)} \be_{\si(l_1),\si^2(l_1)} \ldots \be_{\si^{b-1}(l_1) ,\si^b(l_1)},
\]
where $\si^b(l_1) = l_1$. Moreover, the full product (\ref{beta}) is a product of such loops (if $b \neq k$). To summarise, we write each derangement $\de$ as a product of loops $\ga_1 \ldots \ga_h$.

Summing over all such derangements, it follows that the determinant $E_p$ of $M_p$ is equal to the sum of all loops disjoint from $p$:
\[
E_p = \sum_{k} \sum_{\de \in \De_k} \ga_1 \ldots \ga_{h_{l_k}} ,
\]
where $\De_k$ is the set of derangements of length $k$ in and $l_k$ the number of loops in each such derangement. Since the determinant $D_{n1}$ is  a sum of all products of all paths together with their corresponding subdeterminants $E_p$, it follows that 
\[
D_{n1} = \sum_{p \in \PP_{n,n,1}} \al_p E_p. 
\]
\end{proof}

{\bf Representation of the loops and paths.}
Having seen that the effect of antibiotics on taxon $x_j$ is given by (see (\ref{inta-1}))
\[
\de_1 = -l \frac{c_{j1}}{\det(A)}, 
\]
where $l=\De b_j = K_j \la_j/r_j$ is the antibiotic pressure on $j$ scaled with the intrinsic growth rate $r_j$ of $x_j$,  we can interpret this using Lemma \ref{det-rep}. First we note that since $E_p$ is the sum of all loops {\em disjoint} from the path $p$, we can write 
\[
\det(A) = E_{0},
\]
where $0$ is the empty loop, (i.e. no loops). We conclude that $\det(A)/E_p$ represents all loops which {\em intersect} the path $p$, and thus
\begin{equation} \label{isn-loops}
\de_1 = -l \sum_{p \in \PP_{n,j,1}} \frac{\al_p E_p}{\det(A)} = -l \sum_{p \in \PP_{n,j,1}} \frac{\al_p}{F_p},
\end{equation}
where $F_p = \det(A)/E_p$ represents all loops that intersect $p$. As a special case, we compare it to the results in reference \cite{Sundius}, where $\al_{12}$ is the only non-trivial path from $x_2$ to $x_1$. Since there are no non-trivial loops in dimension $2$, we have that 
\[
\de_1 = -l \frac{\al_{12}}{1-\al_{21}\al_{12}}.
\]
($\al_{21}\al_{12}$ is the only loop that intersects the path $(1,2)$, corresponding to $\al_{12}$). 

{\bf Example}.
Here is another example in dimension $4$. Let $A$ be the interaction matrix as above with elements $(A)_{ij} = \al_{ij}$. As usual assume that $\al_{ii} = 1$, for $1 \leq i \leq 4$. Then the cofactors
\[
c_{31} = \left| \begin{array}{ccc}
\al_{12} &\al_{13} & \al_{14} \\
1 & \al_{23} & \al_{24} \\
\al_{32} & 1 & \al_{34} \end{array}
\right| = \al_{14}(1-\al_{23}\al_{32}) + \al_{12} \al_{23} \al_{34} + \al_{13}\al_{32}\al_{24} - \al_{12}\al_{24} - \al_{13}\al_{34}.
\]
and 
\[
c_{33} = \left| \begin{array}{ccc}
1 & \al_{12} &\al_{13} \\
\al_{21} & 1 & \al_{23} \\
\al_{31} & \al_{32} & 1 \end{array}
\right| = 1 - \al_{12}\al_{21} - \al_{13}\al_{31} - \al_{23}\al_{32} 
+ \al_{12} \al_{23} \al_{31} + \al_{13}\al_{32}\al_{21}.
\]
We see that $c_{33}$ consists of all loops not including $4$. The cofactor $c_{31}$ consists of the paths from $4$ to $1$ term by term multiplied by the loop disjoint from the path. However, only the path $(1,4)$ has a non-trivial disjoint loop, namely $\al_{23} \al_{32}$. The subdeterminant $(1-\al_{23}\al_{32})$ represents the impact of the interaction between taxa $2$ and $3$ on the interaction along the path $(1,4)$, given that the loop $(2,3)$ is not isolated from $(1,4)$. If this is the case, i.e. if $\al_{ij} = \al_{ji} = 0$ for $i \in \{1, 4\}$ and $j \in \{2, 3\}$ 
(this means that $\al_{12} = \al_{21} = \al_{13} = \al_{31} = \al_{24} = \al_{42} = \al_{43} = \al_{34}= 0$), then the qoutient $c_{31}/c_{33}$ boils down to 
\[
\frac{c_{33}}{c_{31}} = \al_{14},
\]
as it should be (i.e. the terms $(1-\al_{23}\al_{32})$ cancel each other). The contribution to the change in the $x_1$-direction is then 
\[
\de_1 = -l \frac{c_{31}}{\det(A)} = -l \frac{\al_{14}(1-\al_{23}\al_{32})}{(1-\al_{23}\al_{32})(1-\al_{14}\al_{41})} = -l \frac{\al_{14}}{1-\al_{14}\al_{41}}.
\]
Again this expression is completely analogous to the results in reference \cite{Sundius}. 
As soon as there is a link between the loop $(2,3)$ and the path $(1,4)$ things will change. Say for example that there is no interaction between $2$ and $4$, and no interaction between $3$ and $4$, and no interaction between $1$ and $2$, (i.e. $\al_{24}=\al_{42} = \al_{34} = \al_{43} = \al_{12} = \al_{21} = 0$), but $\al_{13}, \al_{31} \neq 0$. Then,
\begin{align}
\frac{c_{31}}{\det(A)}  &= \frac{\al_{14} (1 - \al_{23}\al_{32})}{
(1 - \al_{23}\al_{32})(1-\al_{14}\al_{41})  - \al_{13}\al_{31}} \\
&= \frac{\al_{14}} {1-\al_{14}\al_{41} - \frac{\al_{13}\al_{31}}{1-\al_{23}\al_{32}}}.
\end{align}
We see that the interaction between $1$ and $3$ indirectly influences the effect on taxon $x_1$ of antibiotics on taxon $x_4$. In particular, writing $\de_1$ as in (\ref{isn-loops}), we get
\[
F_{14} = 1 - \al_{14}\al_{41} - \frac{\al_{13}\al_{31}}{1-\al_{23}\al_{32}}.
\]

\end{document}